\documentclass[11pt]{amsart}
\usepackage{mathrsfs,amssymb}
\usepackage{verbatim}

\usepackage{graphicx}
\begin{document}

\newtheorem{theorem}{Theorem}[section]
\newtheorem{proposition}[theorem]{Proposition}
\newtheorem{lemma}[theorem]{Lemma}
\newtheorem{corollary}[theorem]{Corollary}
\newtheorem{conjecture}[theorem]{Conjecture}
\newtheorem{question}[theorem]{Question}
\newtheorem{problem}[theorem]{Problem}
\newtheorem{definition}[theorem]{Definition}

\theoremstyle{remark}
\newtheorem{remark}[theorem]{Remark}

\def\theenumi{\roman{enumi}}

\numberwithin{equation}{section}

\renewcommand{\Re}{\operatorname{Re}}
\renewcommand{\Im}{\operatorname{Im}}

\def \R {{\mathbb R}}
\def \HH {{\mathbb H}}
\def \C {{\mathbb C}}
\def \Z {{\mathbb Z}}
\def \Q {{\mathbb Q}}
\def \TT {{\mathbb T}}
\newcommand{\T}{\mathbb T}
\def \Dc {{\mathcal D}}

\newcommand{\tr}[1] {\hbox{tr}\left( #1\right)}

\newcommand{\area}{\operatorname{area}}

\newcommand{\Norm}{\mathcal N}
\newcommand{\simgeq}{\gtrsim}%
\newcommand{\simleq}{\lesssim}

\newcommand{\length}{\operatorname{length}}

\newcommand{\curve}{\mathcal C} 
\newcommand{\vE}{\mathcal E} 
\newcommand{\Ec}{\mathcal {E}} 
\newcommand{\Sc}{\mathcal{S}} 

\newcommand{\dist}{\operatorname{dist}}
\newcommand{\supp}{\operatorname{supp}}
\newcommand{\spec}{\operatorname{spec}}
\newcommand{\diam}{\operatorname{diam}}

\newcommand{\Ccap}{\operatorname{Cap}}
\newcommand{\E}{\mathbb E}

\newcommand{\sumstar}{\sideset{}{^\ast}\sum}

\newcommand {\Zc} {\mathcal{Z}} 
\newcommand{\ninumber}{\Zc}

\newcommand{\zeigen}{E} 
\newcommand{\eigen}{m}

\newcommand{\ave}[1]{\left\langle#1\right\rangle} 

\newcommand{\Var}{\operatorname{Var}}
\newcommand{\Prob}{\operatorname{Prob}}

\newcommand{\var}{\operatorname{Var}}
\newcommand{\Cov}{{\rm{Cov}}}
\newcommand{\meas}{\operatorname{meas}}

\title{Nodal intersections for random eigenfunctions on the torus}
\author{Ze\'ev Rudnick and Igor Wigman}
\address{
School of Mathematical Sciences,
Tel Aviv University, Tel Aviv, Israel} \email{rudnick@post.tau.ac.il}
\address{Department of Mathematics, King's College London, UK}
\email{igor.wigman@kcl.ac.uk}


\begin{abstract}
We investigate the number of nodal intersections of random Gaussian Laplace
eigenfunctions on the standard two-dimensional flat torus
(``arithmetic random waves") with a fixed smooth reference
curve with nonvanishing curvature.  The expected intersection number
is universally proportional to the length of the reference curve,
times the wavenumber, independent of the geometry.

Our main result prescribes the asymptotic behaviour of the nodal
intersections variance for smooth curves in the high energy limit;
remarkably, it is dependent on both the angular distribution of
lattice points lying on the circle with radius corresponding to the
given wavenumber, and the geometry of the given curve. In particular,
this implies that the nodal intersection number admits a universal
asymptotic law with arbitrarily high probability.
\end{abstract}

\date{\today}
\maketitle

\section{Introduction}

\subsection{Background}
A number of recent papers studied the fine structure of nodal lines
of eigenfunctions of the Laplacian, and in particular the number of
intersections of these nodal lines with a fixed reference curve.
Thus let $\curve\subset M$ be a smooth curve on a 
(smooth) Riemannian surface $M$. Let $F$ be a real-valued
eigenfunction of the Laplacian on $M$ with eigenvalue $\lambda^2$:
$-\Delta F =\lambda^2 F$. We want to estimate the number of nodal
intersections
\begin{equation}
\ninumber(F)=\#\{x: F(x)=0\} \cap \curve \end{equation}
 that is the number of zeros of $F$ on $\curve$, as $\lambda\to \infty$.

It is expected that in many  situations, there is an upper bound of
the form $\ninumber(F)\ll \lambda$, and general criteria for this to
happen exist \cite{TZ, ET}, though it is difficult to verify these
criteria in most situations. As for lower bounds, nothing seems to
be known in general, see \cite{GRS} for results on Hecke
eigenfunctions on hyperbolic surfaces (and \cite{Magee} for
analogous results on the sphere), and \cite{Jung, JZ} for results on
density one subsequences for hyperbolic surfaces. Aronovich and
Smilansky \cite{Aronovich Smilansky} studied the nodal intersections
of random monochromatic waves on the plane  ~\cite{Berry 1977}
with various reference curves.


\vspace{3mm}

The one context where we have more information is for the standard
flat torus $\TT^2=\R^2/ \Z^2$. Let $\curve \subset \TT^2$ be a
smooth curve.
Bourgain and Rudnick \cite{BRinvent} showed that if $\curve$ is not
a segment of a closed geodesic, then it is not part of the nodal
line of any eigenfunction with $\lambda>\lambda_\curve$ sufficiently
large, hence $\ninumber(F) <\infty$ for $\lambda$ sufficiently
large. If the reference curve $\curve$ has  nowhere-zero curvature,
they gave upper and lower bounds \cite{BRGAFA} on the intersection
numbers
\begin{equation}
\label{eq:BR ninnumber} \lambda^{1-o(1)} \ll \ninumber(F)\ll \lambda.
\end{equation}
The lower bound is strengthened in \cite{BRNI}, and assuming a
number theoretic conjecture takes the form $\ninumber(F)\gg
\lambda$ and is thus optimal up to a constant multiple. Moreover the
number theoretic condition is known to hold for "generic"
eigenvalues hence we know that for almost all eigenvalues, {\em all}
eigenfunctions in the eigenspace satisfy the lower bound
$\ninumber(F)\gg \lambda$.

In this paper we show that in this setting, for "generic" toral
eigenfunctions there is in fact an asymptotic law for these nodal
intersection numbers.
 We will show that for {\em all} eigenspaces, we in fact have an asymptotic result for "almost all"
eigenfunctions in that eigenspace, once we take a limit of large
eigenspace dimension.

\subsection{Our setting}
Let
\begin{equation}
\vE = \{\mu\in \Z^2:|\mu|^2=\eigen\}
\end{equation}
be the set of lattice points on the circle of radius
$\sqrt{\eigen}$, and denote

\begin{equation} N_\eigen = \#\vE. \end{equation}
We consider the random Gaussian toral eigenfunctions
\begin{equation}
\label{eq:F def} F(x) = \frac{1}{\sqrt{
N_{\eigen}}}\sum\limits_{\mu\in \vE}a_{\mu}e^{2\pi i \langle
\mu,x\rangle},
\end{equation}
with eigenvalue $$ \lambda^2 = 4\pi^2 \eigen,$$ defined on the
standard torus $\TT^2=\R^{2}/\Z^{2}$, where $a_{\mu}$ are standard
complex Gaussian random variables (that is $\E(a_\mu)=0$,
$\E(|a_\mu|^2)=1$), independent save for the relations $a_{-\mu} =
\overline{a}_{\mu}$. The random functions $F$ are called
``arithmetic random waves" ~\cite{KKW}.

We define the probability measures on the unit circle $\mathcal{S}^{1}\subseteq\R^{2}$
\begin{equation}
\label{eq:taum def}
\tau_{\eigen} = \frac{1}{N_{\eigen}}\sum\limits_{\mu\in \vE}\delta_{\mu/\sqrt{\eigen}},
\end{equation}
where $\delta_{x}$ is the Dirac delta function at $x$.
It is well known that the lattice points $\Ec$ are
equidistributed on $\Sc^{1}$ along generic subsequences of energy
levels (see e.g. \cite{FKW}, Proposition 6) in the sense that
$\tau_{m_{j}}\Rightarrow \frac{1}{2\pi}d\theta$ along some density $1$ 
sequence $\{m_{j}\}$ (relatively to the set of integers representable as sum of two squares), 
and thus, in particular, 
$\widehat{\tau_{m_{j}}}(4)\rightarrow 0$. Below we will assume that 
$|\widehat{\tau_{m_{j}}}(4)|\le 1$ is bounded away from $1$ (see the formulation
of the main results);
for $\tau$ invariant w.r.t. rotation by $\pi/2$, $\widehat{\tau}(4)= \pm 1$ if and only if
$\tau = \frac{1}{4}(\sum\limits_{k=0}^{4}\delta_{k\pi/2})$ or
$\tau = \frac{1}{4}(\delta_{\pm \pi/4}+\delta_{\pm 3\pi/4})$
(thinking of the circle as $\Sc^{1}\cong \R/[0,2\pi)$),
thus we only exclude these two limiting probability measures
(see section \ref{sec:leading constant} for more discussion on the possible
limiting angular measures, and the peculiarities of these two).

Given a curve $\curve \subset \TT^2$, we wish to study the
statistics of the number of nodal intersections $\Zc(F)$ for  an
arithmetic random wave  $F$. We do this when the curve $\curve$ is
smooth, with nowhere zero curvature.

\begin{theorem}
\label{thm:expected, variance small} Let $\curve \subset \TT^2$ be a
smooth curve on the torus, with nowhere-zero curvature, of
total length $L$.

 i) The expected number of nodal intersections is precisely
\begin{equation}\label{expected ni}
\E \left[ \ninumber \right] = \sqrt{2\eigen}L =
\frac{\lambda}{\pi\sqrt{2}}L.
\end{equation}

ii) 
Let $\{\eigen\}$ be a sequence s.t. $N_{\eigen}\rightarrow\infty$ and the 
Fourier coefficients $\{\widehat{\tau_{m}}(4)\}$ do not accumulate at $\pm 1$,
i.e. no subsequence of $\{\widehat{\tau_{m}}(4)\}$ converges to $+1$ or $-1$.
Then the variance is
\begin{equation}
\Var( \ninumber ) \ll \frac{\eigen}{N_\eigen}  \ll
\frac{\lambda^2}{N_\eigen}.
\end{equation}
\end{theorem}

By Chebyshev's inequality we deduce that under the conditions of 
Theorem \ref{thm:expected, variance small}, we
have with arbitrarily high probability
\begin{equation}
\ninumber(F) \sim \sqrt{2\eigen}L
\end{equation}
for eigenfunctions with eigenvalue $4\pi^2 \eigen$. Our main result
in fact prescribes the asymptotic form for the variance, which
depends on the distribution of the lattice points $\vE$ once
projected to the unit circle.
\begin{theorem}
\label{thm:nodal var intr asympt} Let $\curve \subset \TT^2$ be a
smooth curve on the torus, with nowhere-zero curvature, of
total length $L$, and $\{\eigen\}$ a sequence s.t. $N_{\eigen}\rightarrow\infty$ and the
Fourier coefficients $\{\widehat{\tau_{m}}(4)\}$ do not accumulate at $\pm 1$.
Then
\begin{equation}
\label{eq:nodal var intr asympt} \Var(\ninumber)  = \left(
4B_\curve(\vE) -L^2 \right) \cdot \frac{\eigen }{N_\eigen}   +
O\left( \frac{\eigen}{N_\eigen^{3/2}}\right)
\end{equation}
where
\begin{equation} \label{def of B in int}
B_\curve(\vE):= \int_{\curve}\int_{ \curve} \frac {1}{N_{\eigen}}
\sum_{\mu \in \vE} \left\langle \frac{\mu}{|\mu|}, \dot \gamma(t_1)
\right\rangle^2 \cdot  \left\langle \frac{\mu}{|\mu|}, \dot
\gamma(t_2) \right\rangle^2 dt_1 dt_2
\end{equation}
with $\gamma:[0,L]\to \curve$ a unit speed parameterization.
\end{theorem}

Theorem~\ref{thm:nodal var intr asympt} immediately implies the
second part of Theorem~\ref{thm:expected, variance small}. In
Section ~\ref{sec:leading constant} we discuss the possible partial
limits of $B_\curve(\vE)$ as $\eigen\to \infty$: there is no unique
limit, similar to what happens for the variance of the length of
nodal lines in this model \cite{KKW}. The leading constant $$0 \le
4B_\curve(\vE) -L^2 \le L^{2}$$ is always non-negative and bounded
(see Proposition \ref{prop:c(mu,gamma)=L^2/4}); it can however
vanish, for instance when $\curve$ is a full circle, see
\S~\ref{sec:leading const circ}.


\subsection{About the proof and plan of the paper}

First, we may restrict $F$ along $\curve$; this reduces computing
the nodal intersections $\ninumber$ to counting zeros of a random
process $f$ defined on an interval. The Kac-Rice
formula (see e.g. ~\cite{CL1967} or \cite[Theorems 11.2.1,
11.5.1]{Adler Taylor}) is a standard tool for studying the expected
number of zeros of a process and its higher moments by expressing
the $k$-th (factorial) moment in terms of a certain $k$-dimensional
integral.

For the expected value of $\Zc$ we do this in
~\S~\ref{sec:expectation}. For the second moment, the Kac-Rice
formula would state
\begin{equation}
\label{eq:Kac-Rice example} \E[\ninumber^{2}] =
\iint\limits_{\curve\times\curve} K_{2}(t_{1},t_{2})dt_{1}dt_{2} +
\E[\ninumber],
\end{equation}
where $K_{2}$ is the suitably defined ``$2$-point correlation
function", that is, provided that we justify its use. Unfortunately,
to our best knowledge, all the available references impose a certain
non-degeneracy condition on $f$ and its derivative $f'$, which is
far from being satisfied. In fact, it is easy to construct an
example where the Kac-Rice integral in \eqref{eq:Kac-Rice example}
is off from computing the second (factorial) moment: one checks
that the functions \eqref{eq:F def} satisfy that $F(x)=0$ if
and only if $F(x+(1/2,1/2))=0$. Hence if $\curve$ is a simple closed
curve, invariant w.r.t. the translation
$$\phi:x\mapsto \left(\frac{1}{2},\frac{1}{2} \right) + x, $$ i.e. $\curve=\curve_{1}\cup\curve_{2}$,
where $\curve_{i}$ are the maximal subsets of $\curve$ so that
$\phi(\curve_{1})=\curve_{2}$, then the total number of nodal
intersections $\ninumber$ is twice the number of intersections with
$\curve_{1}$ (so that the variance is multiplied by $4$); however
the linear part on the RHS of \eqref{eq:Kac-Rice example} is not
invariant, and therefore the precise Kac-Rice formula as stated in
\eqref{eq:Kac-Rice example} is in general wrong.

To cope with this situation we develop an {\em approximate} form of
the Kac-Rice for the second moment of the number of zeros of a
random eigenfunction along a smooth curve, which is
sufficient for our purposes. This is quite delicate and takes up all
of sections ~\ref{sec:2-pnt corr func}, \ref{sec:var(Z_eta)
kac-rice approx}, and  Appendix \ref{app:det>0}; we believe that the developed techniques are of
independent interest, and could be used in a variety of situations
where Kac-Rice is not directly applicable (e.g. ~\cite{CMW}). In our
situation the result gives the variance of $\ninumber$ in terms of
the second moments of the covariance function (also referred to as
covariance {\em kernel}) $r(t_1,t_2)=
\E\{F(\gamma(t_1))F(\gamma(t_2))\}$ and its derivatives $r_j =
\partial r /\partial t_j$, $r_{ij} = \partial^2 r/ \partial
r_i\partial t_j$  along the curve:
\begin{proposition}
\label{prop:var(Z) kac-rice approx} Fix $\epsilon_{0} > 0$.
Then for all $m$ such that $ |\widehat{\tau_{m}}(4)|<1-\epsilon_{0} $ 
one has the following approximate Kac-Rice formula,
\begin{equation}
\label{eq:var(Z) kac-rice approx}
\begin{split}
\Var (\ninumber) =  \eigen &\iint_0^L   \left(  r^{2}-
\left(\frac{r_{1}}{\sqrt{2\pi^2\eigen}}\right)^{2} -
\left(\frac{r_{2}}{\sqrt{2\pi^2\eigen}}\right)^{2}+
\left(\frac{r_{12}}{2\pi^2\eigen}\right)^{2} \right)dt_{1}dt_{2}
\\&+ O\left( \frac{\eigen}{N_{\eigen}^{3/2}}  \right),
\end{split}
\end{equation}
where the implied constant depends only on $\epsilon_{0}$.
\end{proposition}
In the proof of Proposition~\ref{prop:var(Z) kac-rice approx} we
also have to control the fourth moment of $r$ and its derivatives;
this is done in \S~\ref{sec:4th moment r and der along gamma}.

Proposition \ref{prop:var(Z) kac-rice approx} reduces our problem to
evaluating the second moment of the covariance function and its
various derivatives along the given curve. To this end, we
eventually encounter an arithmetic problem, which is to show that
\begin{equation}\label{Riesz energy}
\sum_{\mu\neq \mu'\in \vE} \frac 1{|\mu-\mu'|}  = o(N_\eigen).
\end{equation}
This is done in \S~\ref{sec:r,r1,r2,r12 2nd mom}, appealing among
other things to a theorem of Mordell \cite{Mordell 1932} about
representing a binary quadratic form as a sum of two squares, in
other words counting the number of pairs of distinct vectors $(\mu,
\mu') \in \vE\times \vE$ with a given inner product. The
$3$-dimensional version of the quantity \eqref{Riesz energy} is
essentially the electrostatic energy of point charges placed at the
integer points at on the sphere of radius $\sqrt{\eigen}$ and is
analyzed in \cite{BRS}.

The  term $B_\curve(\vE)$ in \eqref{def of B in int}, which determines
the leading term of the variance, arises from the asymptotics
of the second moment $\iint (r_{12})^2$. In \S~\ref{sec:leading
constant} we analyze $B_\curve(\vE)$ and determine when it vanishes
and its limiting value distribution when $N_\eigen\to \infty$, as a
function of the curve $\curve$.



\subsection{Acknowledgements}

We thank the Israel Institute for Advanced Studies of
Jerusalem for its hospitality during the writing of this paper.
We would like to thank Domenico Marinucci and Valentina Cammarota
for several discussions.
The research leading to these results has received funding from the
European Research Council under the European Union's Seventh
Framework Programme (FP7/2007-2013) / ERC grant agreements
n$^{\text{o}}$ 320755 (Z.R.) and  n$^{\text{o}}$ 335141 (I.W.), and
by an EPSRC Grant EP/J004529/1 under the First Grant Scheme (I.W.).

\section{The expected number of nodal intersections}\label{sec:expectation}

\subsection{Kac-Rice formula for computing the expected number of zeros}

\label{eq:Kac-Rice expectation}

Let $f:I\rightarrow \R$ be a centered Gaussian random function
(``process"), a.s. smooth (e.g. $C^{2}$), with the parameter space
$I$ some nice subset of $\R$, e.g. a closed interval or a finite
collection of closed intervals, and let
$$r(t_{1},t_{2})=r_{f}(t_{1},t_{2}):=\E [f(t_{1})f(t_{2})]$$ be the
covariance function of $f$. Denote $\Zc$ to be the number of zeroes
of $f$ on $I$. For $t\in I$ define $K_{1}(t)=K_{1;f}(t)$ to be the
Gaussian expectation
\begin{equation*}
K_{1}(t) = \phi_{f(t)}(0) \cdot \E[|f'(t)| \big| f(t)=0],
\end{equation*}
where $\phi_{f(t)}$ is the probability density function of the
random variable $f(t)$. The latter involves the centered Gaussian
vector $(f(t),f'(t))$ with covariance matrix
\begin{equation*}
\Gamma(t)=\Gamma_{f}(t)=\left( \begin{matrix} r(t,t) & \partial_{t_{1}}r(t_{1},t_{2})|_{(t,t)} \\
\partial_{t_{1}}r(t_{1},t_{2})|_{(t,t)} & \partial_{t_{1}}\partial_{t_{2}}r(t_{1},t_{2})|_{(t,t)} \end{matrix}\right).
\end{equation*}

The function $K_{1}(t)$ is the {\em zero density} (or first
intensity) of $f$; it may be computed explicitly in terms of the
entries of the matrix $\Gamma(t)$, and in our case the expression is
especially simple, as $\Gamma(t)$ is diagonal and independent of $t$
(a consequence of the fact that our process is induced from an
underlying $2$-dimensional {\em stationary} field restricted on a
curve), see below. By the Kac-Rice formula, if for all $t\in I$ the
matrix $\Gamma(t)$ is nonsingular, then ~\cite{CL1967}
$$\E[\Zc] = \int\limits_{I}K_{1}(t)dt.$$

\subsection{Zero density for nodal intersections}

The random field $F(x)$ is centered Gaussian with covariance
function
\begin{equation*}
r_{F}(x,y):= \E[F(x)\cdot
F(y)]=\frac{1}{N_{\eigen}}\sum\limits_{\mu\in \vE}\cos(2\pi  \langle
\mu,y-x\rangle)
\end{equation*}
for $x,y\in \TT^2$; it is {\em stationary} in the sense that
$r_{F}(x,y)=r_{F}(y-x)$ depends on $y-x$ only (by the well-accepted
abuse of notation). Let $\gamma(t):[0,L]\rightarrow \TT^2$ be the
arc-length parameterization of $\curve$; it induces the process
\begin{equation}
\label{eq:fm=Fm induced} f(t)=F(\gamma(t))
\end{equation}
on $I:=[0,L]$ with the covariance function
\begin{equation*}
r(t_{1},t_{2})=r_{F}(\gamma(t_{1})-\gamma(t_{2}));
\end{equation*}
the process $f$ is unit variance. Let $\Zc$ be the number of zeros
of $f$ (on $I$); it equals the number of nodal intersections of $F$
with $\curve$.

\begin{lemma}
\label{lem K1} The zero density of $f$ is
\begin{equation*}
K_{1}(t)=K_{1;m}(t)\equiv \sqrt{2}\cdot \sqrt{m}.
\end{equation*}
In particular,
\begin{equation*}
\E[\Zc] = \sqrt{2}\sqrt{m}\cdot L.
\end{equation*}
\end{lemma}

To facilitate the computation of the zero density we formulate the
following lemma whose proof will be given in a moment. It is
probably well-known, but nevertheless we give it here as we didn't
find a direct reference.

\begin{lemma}
\label{lem:unit var indep} If $f$ is unit variance, then for every
$t\in [0,L]$, $f(t)$ is independent of $f'(t)$.
\end{lemma}

\begin{proof}[Proof of Lemma \ref{lem K1} assuming Lemma \ref{lem:unit var indep}]

We are to compute the zero density of $f(t)$:
\begin{equation}
\label{eq:K1 Gauss exp} K_{1}(t)=\frac{1}{\sqrt{2\pi}} \E[|f'(t)|
\big| f(t)=0],
\end{equation}
thus we are to compute the covariance matrix of $(f(t),f'(t))$.
Since $f$ is unit variance, by Lemma \ref{lem:unit var indep}, the
covariance matrix is
\begin{equation*}
A_{m}= \left( \begin{matrix} 1 &\\ & \alpha
\end{matrix}\right),
\end{equation*}
where
\begin{equation}
\label{eq:alpha2 def}
\alpha=\alpha_{m}(t)=\frac{\partial^{2}}{\partial t_{1}\partial
t_{2}}r|_{(t,t)},
\end{equation}
and, upon computing the Gaussian expectation \eqref{eq:K1 Gauss exp}
explicitly (see e.g. ~\cite{CL1967}), we obtain
\begin{equation}
\label{eq:K1m=1/pi*sqrt*lambda2} K_{1;m}(x) =
\frac{1}{\pi}\sqrt{\alpha}.
\end{equation}

Now (chain rule)
\begin{equation}
\label{eq:r1}
\partial_{t_{1}}r(t_{1},t_{2}) = \nabla r_{F}(\gamma(t_{1})-\gamma(t_{2})) \cdot  \dot{\gamma}(t_{1})
\end{equation}
and
\begin{equation*}
\alpha = -\dot{\gamma}(t_{2})^{t}\cdot
H_{r_{F}}(\gamma(t_{1})-\gamma(t_{2}))\cdot
\dot{\gamma}(t_{1})|_{(t,t)},
\end{equation*}
where $H_{r_{F}}$ is the Hessian of $r_{F}$ (thought of as
$r_{F}(x)=r_{F}(x_{1},x_{2})$. The Hessian $H_{r_{F}}(0)$ was
computed to be a scalar matrix ~\cite{RW}
\begin{equation*}
H_{r_{F}}(0) = - 2\pi^{2}m \cdot I_{2},
\end{equation*}
so that universally
\begin{equation}
\label{eq:alpha2 explicit} \alpha = 2\pi^{2}m \|\dot{\gamma}(t)
\|^{2} = 2\pi^{2}m ,
\end{equation}
since we assumed that $t$ is the arc-length parameter of $\curve$
(i.e. $\|\dot{\gamma}(t) \|=1$), and the zero density is
\begin{equation*}
K_{1}(t) =  \sqrt{2}\cdot \sqrt{m}.
\end{equation*}

\end{proof}

\begin{proof}[Proof of Lemma \ref{lem:unit var indep}]
The correlation between $f(t)$ and $f'(t)$ is given by
\begin{equation*}
\E[f(t)f'(t)] = \frac{\partial}{\partial t_{1}}r|_{(t,t)}.
\end{equation*}
Since we know that
\begin{equation*}
r(t,t) = 1,
\end{equation*}
upon differentiating,
\begin{equation*}
0= \left(\frac{\partial}{\partial t_{1}}r+\frac{\partial}{\partial
t_{2}}r\right)|_{(t,t)} = 2\frac{\partial}{\partial t_{1}}r|_{(t,t)}
\end{equation*}
by the symmetry.
\end{proof}

\begin{remark}
In fact, the proof above shows that the covariance of the
underlying stationary field $F$ satisfies $\nabla r_{F}(0) = 0$, as
$r_{F}(x,x)\equiv 1$.

\end{remark}

\section{The $2$-point correlation function}

\label{sec:2-pnt corr func}

\subsection{Kac-Rice formula for computing the second moment of the number of zero crossings}

\label{sec:Kac-Rice var}

Let $f$ and $\Zc$ be as in section \ref{eq:Kac-Rice expectation}. We
define the $2$-point correlation function $K_{2}=K_{2;f}:I\times
I\rightarrow \R$ (also called the second intensity) in the following way:
for $t_{1}\ne t_{2}$ we define it as the conditional Gaussian expectation
\begin{equation*}
K_{2}(t_{1},t_{2}) = \phi_{t_{1},t_{2}}(0,0)\cdot
\E[|f'(t_{1})|\cdot |f'(t_{2})| |f(t_{1})=f(t_{2})=0]
\end{equation*}
where $\phi_{t_{1},t_{2}}$ is the probability density function of
the random Gaussian vector $(f(t_{1}),f(t_{2}))$. The function
$K_{2}$ admits a continuation to a smooth function on the whole of $I\times I$
(see \ref{sec:K2 diag =O(m)}), though its values at the diagonal are
of no significance for our purposes.
We will find an explicit expression for $K_{2}(t_{1},t_{2})$ in terms of $r$
and its derivatives (see Lemma \ref{lem:K2 explicit} below); finding
such an expression involves studying the centered Gaussian vector
$(f(t_{1}),f(t_{2}),f'(t_{1}),f'(t_{2}))$ with the covariance matrix
$\Sigma=\Sigma_{4\times 4}(t_{1},t_{2})$ as in \eqref{eq:Sigma def}.

It is known ~\cite{CL1967} that under the assumption that for all
$t_{1}\ne t_{2}$ the matrix $\Sigma(t_{1},t_{2})$ is nonsingular
(i.e. the Gaussian distribution of
$$(f(t_{1}),f(t_{2}),f'(t_{1}),f'(t_{2}))$$ is nondegenerate), the
factorial second moment of $\Zc$ is
\begin{equation*}
\E[\Zc^{2}-\Zc] = \iint\limits_{I\times I}
K_{2}(t_{1},t_{2})dt_{1}dt_{2},
\end{equation*}
so that accordingly
\begin{equation}
\label{eq:Kac-Rice cov}
\var(\Zc) = \int\limits_{I\times I} \left(
K_{2}(t_{1},t_{2})-K_{1}(t_{1})\cdot K_{1}(t_{2})\right)dt_{1}dt_{2}
+ \E[\Zc];
\end{equation}
note that the ``extra" $\E[\Zc]$ manifests the degeneracy of the
matrix $\Sigma(t_{1},t_{2})$ on the diagonal $t_{2}=t_{1}$.

Moreover, if $I_{1},I_{2}\subseteq I$ are {\em disjoint} nice sets
(e.g. intervals), and the degeneracy assumption holds for all
$(t_{1},t_{2}) \in I_{1}\times I_{2}$, then if for $J\subseteq I$ we
denote $\Zc_{J}$ to be the number of zero crossing of $f$ in $J$,
then (either employing the proof in ~\cite{CL1967} or using
\cite[Theorems 11.2.1, 11.5.1]{Adler Taylor} on $I_{1}\cup I_{2}$,
whence we will need to make the non-degeneracy assumption for all
$(t_{1},t_{2}) \in (I_{1}\cup I_{2})^{2}$)
\begin{equation*}
\E[\Zc_{I_{1}}\cdot \Zc_{I_{2}}] = \int\limits_{I_{1}\times I_{2}}
K_{2}(t_{1},t_{2})dt_{1}dt_{2},
\end{equation*}
so that
\begin{equation}
\label{eq:CovIiIj disjoint}
\Cov[\Zc_{I_{1}}\cdot \Zc_{I_{2}}] = \int\limits_{I_{1}\times I_{2}}
\left( K_{2}(t_{1},t_{2}) -
K_{1}(t_{1})K_{1}(t_{2})\right)dt_{1}dt_{2}.
\end{equation}

However, the non-degeneracy assumption is not satisfied in the
case of $f$ as in \eqref{eq:fm=Fm induced}, and we may construct
examples of curves, where the Kac-Rice formula as stated is wrong.
However, in a situation like this we
will be able to write an {\em approximate} Kac-Rice formula,
prescribing the same order of magnitude for the fluctuations of the
nodal intersections as the {\em precise} Kac-Rice (see Proposition
\ref{prop:var(Z) kac-rice approx}). We will see in section
\ref{sec:K2m expansion r moments} (Proposition \ref{prop:K2m expansion r moments})
that under certain conditions on $r$ (namely
that $|r|$ is bounded away from $1$) we will be able to approximate
the $2$-point correlation function in terms of powers of $r$ and its
derivatives; this will allow us to write the approximate Kac-Rice
formula of Proposition \ref{prop:var(Z) kac-rice approx} in
terms of the relevant moments of $r$ and its derivatives rather than
in terms of the integral of $2$-point correlation function. We will
prove the approximate Kac-Rice formula of Proposition
\ref{prop:var(Z) kac-rice approx} in section
\ref{sec:var(Z_eta) kac-rice approx} assuming the preparatory work in section
\ref{sec:K2m expansion r moments}, and some upper bounds for the
$4$-th moments of $r$ and its derivatives along the relevant curve
in section \ref{sec:4th moment r and der along gamma} (Lemma
\ref{lem:4th moment r and der along gamma}).

\subsection{An explicit expression for the $2$-point correlation function}

Let $K_{2}(t_{1},t_{2})=K_{2;m}(t_{1},t_{2})$ be the $2$-point
correlation function of our process $f$ as in \eqref{eq:fm=Fm
induced}, i.e. for $t_{2}\ne t_{1}$
\begin{equation*}
K_{2}(t_{1},t_{2}) = \phi_{t_{1},t_{2}}(0,0)\cdot
\E[|f'(t_{1})|\cdot |f'(t_{2})| |f(t_{1})=f(t_{2})=0],
\end{equation*}
where $\phi_{t_{1},t_{2}}$ is the probability density function of
the random Gaussian vector $(f(t_{1}),f(t_{2}))$. The following
lemma gives an explicit expression for $K_{2}$ in terms of $r_{f}$
and its derivatives; recall the definition \eqref{eq:alpha2
def} for $\alpha$ and its explicit value $\alpha=2\pi^{2}\eigen$.

\begin{lemma}
\label{lem:K2 explicit} We have explicitly
\begin{equation}
\label{eq:K2 explicit} K_{2}=K_{2;m}(t_{1},t_{2}) = \frac{1}{\pi^{2}
(1-r^{2})^{3/2}}\cdot \mu\cdot
(\sqrt{1-\rho^{2}}+\rho\arcsin{\rho}),
\end{equation}
where 
\begin{equation}
\label{eq: mu def} \mu = \mu_{\eigen}(t_{1},t_{2}) =
\sqrt{\alpha(1-r^{2})-r_{1}^{2}}\cdot \sqrt{\alpha
(1-r^{2})-r_{2}^{2}},
\end{equation}
and
\begin{equation}
\label{eq:rho def} \rho = \rho_{m}(t_{1},t_{2}) =
\frac{r_{12}(1-r^{2})+rr_{1}r_{2}}{\sqrt{\alpha (1-r^{2})-r_{1}^{2}}
\cdot \sqrt{\alpha (1-r^{2})-r_{2}^{2}}},
\end{equation}
is the correlation between the derivatives $f'(t_{1})$ and
$f'(t_{2})$, conditioned on both values vanishing (thus satisfying
$|\rho|\le 1$).
\end{lemma}

\begin{proof}

The covariance matrix for
$(f_{m}(t_{1}),f_{m}(t_{2}),f_{m}'(t_{1}),f_{m}'(t_{2}))$ is
\begin{equation}
\label{eq:Sigma def} \Sigma = \left( \begin{matrix} A & B \\ B^{t}
&C
\end{matrix}\right),
\end{equation}
where
\begin{equation*}
A= \left(\begin{matrix} 1 &r \\ r & 1
\end{matrix} \right), \quad
B = \left(\begin{matrix} 0 & \frac{\partial r }{\partial t_{2}}  \\
\frac{\partial r }{\partial t_{1}}   & 0\end{matrix} \right), \quad
C = \left( \begin{matrix} \alpha  & \frac{\partial^{2}r}{\partial
t_{1}\partial t_{2}}
\\ \frac{\partial^{2}r }{\partial t_{1}\partial t_{2}} &\alpha
\end{matrix}  \right).
\end{equation*}
We abbreviate
$$
r_{1}:= \frac{\partial r}{\partial t_{1}} , \quad r_{2}:=
\frac{\partial r}{\partial t_{2}} , \quad r_{12}:=
\frac{\partial^{2}r}{\partial t_{1}\partial t_{2}} .
$$
The covariance matrix for the conditional distribution of
$f_{m}'(t_{1}),f_{m}'(t_{2})$ conditioned on
$f_{m}(t_{1})=f_{m}(t_{2})=0$ is
\begin{equation}
\label{eq:Omega def}
\begin{split}
\Omega&=\Omega_{m}(t_{1},t_{2}) = C - B^{t}A^{-1}B = \left(
\begin{matrix} \alpha  &  r_{12}
\\ r_{12} &\alpha
\end{matrix}  \right)-\frac{1}{1-r^{2}}\left( \begin{matrix}
r_{1}^{2} &-rr_{1}r_{2} \\ -rr_{1}r_{2} &r_{2}^{2}
\end{matrix}   \right) \\&= \frac{1}{1-r^{2}}\left(\begin{matrix}
\alpha(1-r^{2})-r_{1}^{2} & r_{12}(1-r^{2})+rr_{1}r_{2} \\
r_{12}(1-r^{2})+rr_{1}r_{2} & \alpha(1-r^{2})-r_{2}^{2}
\end{matrix} \right).
\end{split}
\end{equation}
The two-point correlation function is then given by
\begin{equation*}
K_{2;m}(t_{1},t_{2}) = \frac{1}{2\pi \sqrt{\det{A}}}\E
[|W_{1}W_{2}|],
\end{equation*}
where $$(W_{1},W_{2})\sim N(0,\Omega)$$ are centered Gaussian with
covariance $\Omega$. By normalizing the random variables
$$(W_{1},W_{2}) =
\left(\frac{\sqrt{\alpha(1-r^{2})-r_{1}^{2}}}{\sqrt{1-r^{2}}}Y_{1},
\frac{\sqrt{\alpha(1-r^{2})-r_{2}^{2}}}{\sqrt{1-r^{2}}}Y_{2}\right)$$
we write $K_{2;m}$  as
\begin{equation}
\label{eq:K2 form} K_{2;m} = \frac{1}{2\pi (1-r^{2})^{3/2}}\cdot
\mu\cdot \E [|Y_{1}Y_{2}|],
\end{equation}
where $\mu$ is given by \eqref{eq: mu def}, $(Y_{1},Y_{2})\sim
N(0,\Delta(\rho))$ with
\begin{equation}
\label{eq:Delta(rho) def} \Delta(\rho) = \left( \begin{matrix} 1
&\rho \\ \rho &1
\end{matrix}\right)
\end{equation}
and $\rho$ is given by \eqref{eq:rho def}.

It remains to evaluate
\begin{equation*}
G(\rho)=\E[|Y_{1}Y_{2} |]
\end{equation*}
with $(Y_{1},Y_{2})\sim N(0,\Delta(\rho))$. We may compute $G$
explicitly to be equal to (see e.g. \cite{Bleher-Di}),
\begin{equation}
\label{eq:G(rho) def} G(\rho) =  \frac{2}{\pi}\left(
\sqrt{1-\rho^{2}}+\rho\arcsin\rho \right),
\end{equation}
which finally yields the explicit formula \eqref{eq:K2 explicit} via
\eqref{eq:K2 form}.
\end{proof}

\subsection{Asymptotics for the $2$-point correlation function}

\label{sec:K2m expansion r moments}

\begin{proposition}
\label{prop:K2m expansion r moments}

For every $\epsilon_{2}>0$, the two point correlation function satisfies,
uniformly for $|r| < 1-\epsilon_{2}$:
\begin{equation}
\label{eq:K2m expansion r moments}
\begin{split}
K_{2}(t_{1},
t_{2})&=\frac{\alpha}{\pi^{2}}\left(1+\frac{1}{2}r^{2}-\frac{1}{2}(r_{1}/\sqrt{\alpha})^{2}
-\frac{1}{2}(r_{2}/\sqrt{\alpha})^{2}
+\frac{1}{2}(r_{12}/\alpha)^{2} \right)
\\&+\alpha \cdot O\left(r^{4}+(r_{1}/\sqrt{\alpha})^{4}+
(r_{2}/\sqrt{\alpha})^{4} + (r_{12}/\alpha)^{4}\right).
\end{split}
\end{equation}
Bearing in mind \eqref{eq:K1m=1/pi*sqrt*lambda2}, we may
equivalently write
\begin{equation*}
\begin{split}
K_{2}(t_{1}, t_{2})- K_{1}(t_{1})K_{1}(t_{2})& =
\frac{\alpha}{2\pi^{2}}\left( 
 r^{2}- \left(\frac{r_{1}}{\sqrt{\alpha}}\right)^{2}
- \left(\frac{r_{2}}{\sqrt{\alpha}}\right)^{2} +
\left(\frac{r_{12}}{\alpha}\right)^{2} \right)
\\&+\alpha \cdot O\left(r^{4}+\left(\frac{r_{1}}{\sqrt{\alpha}}\right)^{4}+
\left(\frac{r_{2}}{\sqrt{\alpha}}\right)^{4} +
\left(\frac{r_{12}}{\alpha}\right)^{4}\right)
\end{split}
\end{equation*}
with constants involved in the `O`-notation depending on $\epsilon_{2}$
only.

\end{proposition}

\begin{proof}
Note that if $r$, $\frac{r_{1}}{\sqrt{m}}$,
$\frac{r_{2}}{\sqrt{m}}$, and $\frac{r_{12}}{m}$ are {\em small},
then $\rho$ is small too. We may then expand $\rho$ and $\mu$ about
$r =0$, $\frac{r_{1}}{\sqrt{\alpha}}=0$,
$\frac{r_{2}}{\sqrt{\alpha}}=0$, $\frac{r_{12}}{\alpha}=0$:
\begin{equation}
\label{eq:rho Taylor}
\begin{split}
\rho &= \frac{r_{12}}{\alpha} \cdot
\left(1-(r^{2}+(r_{1}/\alpha)^{2})\right)^{-1/2}) \cdot
\left(1-(r^{2}+(r_{2}/\alpha)^{2})\right)^{-1/2}) \\&+
O\left(r^{3}+(r_{1}/\sqrt{\alpha})^{3}+ (r_{2}/\sqrt{\alpha})^{3} +
(r_{12}/\alpha)^{3}\right)
\\&= \frac{r_{12}}{\alpha} + O\left(r^{3}+(r_{1}/\sqrt{\alpha})^{3}+
(r_{2}/\sqrt{\alpha})^{3} + (r_{12}/\alpha)^{3}\right),
\end{split}
\end{equation}

Next we need to Taylor expand the function $G(\rho)$ as in
\eqref{eq:G(rho) def} about $\rho=0$:
\begin{equation*}
G(\rho)
= \frac{2}{\pi}\left(1+ \frac{1}{2}\rho^{2}\right) + O(\rho^{4}).
\end{equation*}
Substituting \eqref{eq:rho Taylor}, we obtain
\begin{equation*}
G(\rho) = \frac{2}{\pi}\left(1+
\frac{1}{2}(r_{12}/\alpha)^{2}\right) +
O\left(r^{4}+(r_{1}/\sqrt{\alpha})^{4}+ (r_{2}/\sqrt{\alpha})^{4} +
(r_{12}/\alpha)^{4}\right).
\end{equation*}

Next,
\begin{equation*}
\begin{split}
\mu &= \alpha\sqrt{1-(r^{2}+(r_{1}/\sqrt{\alpha})^{2})}\cdot
\sqrt{1-(r^{2}+(r_{2}/\sqrt{\alpha})^{2})}
\\&= \alpha\left(1-r^{2} - \frac{1}{2}(r_{1}/\sqrt{\alpha})^{2} -\frac{1}{2}(r_{2}/\sqrt{\alpha})^{2}
\right) \\&+\alpha O\left(r^{4}+(r_{1}/\sqrt{\alpha})^{4}+
(r_{2}/\sqrt{\alpha})^{4} + (r_{12}/\alpha)^{4}\right),
\end{split}
\end{equation*}
and  
\begin{equation*}
\frac{1}{(1-r^{2})^{3/2}} = 1 + \frac{3}{2}r^{2}+O(r^{4}).
\end{equation*}

Finally, substituting all the estimates above into \eqref{eq:K2
form} we obtain
\begin{equation*}
\begin{split}
K_{2;m}(t_{1}, t_{2})& = \frac{1}{2\pi} \cdot \left(1 +
\frac{3}{2}r^{2} \right) \cdot \alpha\left(1-r^{2} -
\frac{1}{2}\left(\frac{r_{1}}{\sqrt{\alpha}}\right)^{2}
-\frac{1}{2}\left(\frac{r_{2}}{\sqrt{\alpha}}\right)^{2} \right)\\ & \quad
\cdot \frac{2}{\pi}\left(1+ \frac{1}{2}(\frac{r_{12}}{\alpha})^{2}\right) +\alpha
O\left(r^{4}+\left(\frac{r_{1}}{\sqrt{\alpha}}\right)^{4}+
\left(\frac{r_{2}}{\sqrt{\alpha}}\right)^{4} +
\left(\frac{r_{12}}{\alpha}\right)^{4}\right)
\\&= \frac{\alpha}{\pi^{2}}\left(1+\frac{1}{2}r^{2}-\frac{1}{2}\left(r_{1}/\sqrt{\alpha}\right)^{2}
-\frac{1}{2}\left(r_{2}/\sqrt{\alpha}\right)^{2}
+\frac{1}{2}\left(r_{12}/\alpha\right)^{2} \right)
\\&\quad +\alpha O\Big(r^{4}+(r_{1}/\sqrt{\alpha})^{4}+
(r_{2}/\sqrt{\alpha})^{4} + (r_{12}/\alpha)^{4}\Big).
\end{split}
\end{equation*}
An inspection of each step reveals that all the expansions are valid
under the assumption that $|r|$ is bounded away from $1$.
\end{proof}

\section{Approximate Kac-Rice for computing the variance of nodal intersections}

\label{sec:var(Z_eta) kac-rice approx}

This section is entirely dedicated
to proving Proposition \ref{prop:var(Z) kac-rice approx}. Throughout the present
section we assume that $\epsilon_{0}>0$ is fixed, and $\eigen$ satisfies $|\widehat{\tau_{\eigen}}(4)| < 1-\epsilon_{0}$.

\subsection{Nodal intersections on short arcs}

\label{sec:short arcs}

Let $c_{0}>0$ be a small number (depending on $\epsilon_{0}$),
and divide our curve into short arcs
of size roughly $\frac{c_{0}}{\sqrt{\eigen}}$.
More precisely, let
$K=K_{\eigen}=\left\lfloor L \cdot \frac{\sqrt{\eigen}}{c_{0}} \right\rfloor + 1$,
$$\delta_{0} =\delta_{0;\eigen}= \frac{L}{K} \le \frac{c_{0}}{\sqrt{\eigen}},$$
and define the partition $I=\bigcup\limits_{i=1}^{K}I_{i}$ of $I=[0,L]$
into short intervals $$I_{i}:= [(i-1)\cdot \delta_{0},i\cdot \delta_{0}],$$ $i=1,\ldots, K$,
disjoint save for the overlaps at the endpoints. 
We will eventually choose $c_{0}$ sufficiently small so that the Kac-Rice formula will
hold on the short intervals (see Lemma \ref{lem:matrix nongen diag}), and the value of $r$ or of one of its derivatives in a ``singular 
cube" will be bounded away from $0$ 
(see Definition \ref{def:singular} and Lemma \ref{lem:on B either r or der big}).

For the future we record that, as $c_{0}>0$
is constant,
\begin{equation}
\label{eq:delta<<>>1/sqrt(m)}
\delta_{0} \asymp \frac{1}{\sqrt{\eigen}}.
\end{equation}

For $1\le i \le K$, let $\Zc_{i}$ be the number of nodal
intersections of $F_{m}$ with $\gamma(I_{i}),$ that is
$\Zc_{i}$ is the number of zeros of $f$ on $I_{i}$. We have a.s.
\begin{equation*}
\mathcal{Z} = \sum\limits_{i=1}^{K}\mathcal{Z}_{i},
\end{equation*}
so that
\begin{equation}
\label{eq:tot 2nd mom = sum 2nd mom} \E[\mathcal{Z}^{2}] =
\sum\limits_{i=1}^{K}\E[\mathcal{Z}_{i}^{2}]
+2\sum\limits_{i<j}\E[\mathcal{Z}_{i}\cdot
\mathcal{Z}_{j}];
\end{equation}
equivalently
\begin{equation}
\label{eq:tot var = sum covar} \var(\mathcal{Z}) =
\sum\limits_{i=1}^{K}\var(\mathcal{Z}_{i})
+2\sum\limits_{i<j}Cov\left(\mathcal{Z}_{i},
\mathcal{Z}_{j}\right).
\end{equation}

Later we will apply Kac-Rice \eqref{eq:Kac-Rice cov} to ``most" of the summands in \eqref{eq:tot var = sum covar} (see section
\ref{sec:var(Zi) kac-rice approx}) and bound the contribution of the rest of the summands;
integrating and summing these up will eventually establish the statement of Proposition \ref{prop:var(Z) kac-rice approx}.

\subsection{Nodal intersections variance on short arcs}

As a first goal, we will establish an estimate on the variance $\var(\mathcal{Z}_{i})$ of nodal intersections with
a short arc of $\gamma$; with the help of the latter we will be able to control the contribution
of any individual summand in \eqref{eq:tot var = sum covar}, via Cauchy-Schwartz
(Corollary \ref{cor:Cov(Zi,Zj)=O(1)}).

\begin{proposition}
\label{prop:var short = O(1)}
For every $0<\epsilon_{0}<1$ we can choose $c_{0}=c_{0}(\epsilon_{0})$ sufficiently small, such that
for any $\eigen$ with $|\widehat{\tau_{\eigen}}(4)| < 1-\epsilon_{0}$, we have
\begin{equation*}
\var(\mathcal{Z}_{i}) = O(1),
\end{equation*}
uniformly for $i\le K$,
where the constant involved in the ``O"-notation depends on $\epsilon_{0}$ and $c_{0}$ only.
\end{proposition}

Before proving Proposition \ref{prop:var short = O(1)} we draw the following corollary, as
announced above.

\begin{corollary}
\label{cor:Cov(Zi,Zj)=O(1)}
For every $0<\epsilon_{0}<1$ we can choose $c_{0}=c_{0}(\epsilon_{0})$ sufficiently small, such that
for any $\eigen$ with $|\widehat{\tau_{\eigen}}(4)| < 1-\epsilon_{0}$, we have
\begin{equation*}
Cov(\Zc_{i},\Zc_{j}) = O(1),
\end{equation*}
uniformly for $i,j\le K$,
where the constant involved in the ``O"-notation depends on $\epsilon_{0}$ and $c_{0}$ only.
\end{corollary}

\begin{proof}[Proof of Corollary \ref{cor:Cov(Zi,Zj)=O(1)}]

Applying Cauchy-Schwartz we have
\begin{equation*}
Cov(\Zc_{i},\Zc_{j}) \le \sqrt{\var(\mathcal{Z}_{i}) \cdot \var(\mathcal{Z}_{j})}=O(1),
\end{equation*}
by Proposition \ref{prop:var short = O(1)}.

\end{proof}

To prove Proposition \ref{prop:var short = O(1)} we will need Lemma \ref{lem:matrix nongen diag}
and Proposition \ref{prop:K2 diag =O(m)} stated below.

\begin{lemma}
\label{lem:matrix nongen diag}

For every $0<\epsilon_{0}<1$ we can choose $c_{0}=c_{0}(\epsilon_{0})$ sufficiently small, such that
for any $\eigen$ with $|\widehat{\tau_{\eigen}}(4)| < 1-\epsilon_{0}$, the
matrix $\Sigma(t_{1},t_{2})$, defined in \eqref{eq:Sigma def}, is nonsingular for all $t_{1},t_{2} \in [0,L]^{2}$
with $$0 < |t_{2}-t_{1}| < \frac{c_{0}}{\sqrt{\eigen}}.$$

\end{lemma}

The proof of Lemma \ref{lem:matrix nongen diag} is quite long and technical,
and is thereupon relegated to Appendix \ref{app:det>0}.

\begin{proposition}
\label{prop:K2 diag =O(m)} For $t_{1}\in [0,L]$ and $|t_{2}-t_{1}| <
\frac{c_{0}}{\sqrt{m}}$ one has the uniform estimate
\begin{equation*}
K_{2}(t_{1},t_{2}) = O(\eigen)
\end{equation*}
with constant depending on $c_{0}$ only.
\end{proposition}

The proof of Proposition \ref{prop:K2 diag =O(m)} is deferred to
section \ref{sec:K2 diag =O(m)}.

\begin{proof}[Proof of Proposition \ref{prop:var short = O(1)} assuming Lemma \ref{lem:matrix nongen diag}
and Proposition \ref{prop:K2 diag =O(m)}]

Thanks to Lemma \ref{lem:matrix nongen diag} the covariance matrix
$\Sigma(t_{1},t_{2})$ is nonsingular for all $(t_{1},t_{2})\in I_{i}^{2}$
with $t_{2}\ne t_{1}$, so, by the discussion in section \ref{sec:Kac-Rice var} above we may apply Kac-Rice
\eqref{eq:Kac-Rice cov} to $I_{i} \subseteq I$ to write
\begin{equation}
\label{eq:var Ii Kac-Rice}
\var(\Zc_{i}) = \int\limits_{I_{i}\times I_{i}}
(K_{2}(t_{1},t_{2})-K_{1}(t_{1})K_{1}(t_{2}))dt_{1}dt_{2} + \E[\Zc_{i}].
\end{equation}
Applying Proposition \ref{prop:K2 diag =O(m)} and the Kac-Rice formula \eqref{eq:Kac-Rice expectation}
for computing the expected number of zeros on $I_{i}$
\begin{equation*}
\E[\Zc_{i}] = \int\limits_{I_{i}} K_{1}(t)dt \ll  \sqrt{\eigen} \cdot \delta_{0}
\end{equation*}
(see Lemma \ref{lem K1}) to \eqref{eq:var Ii Kac-Rice} yields
\begin{equation*}
\var(\Zc_{i}) \ll \eigen \cdot \delta_{0}^{2} + \sqrt{\eigen} \cdot \delta_{0} \ll 1,
\end{equation*}
bearing in mind \eqref{eq:delta<<>>1/sqrt(m)}. This concludes the proof of the present proposition.

\end{proof}

\subsection{Proof of Proposition \ref{prop:var(Z) kac-rice approx}}

\label{sec:var(Zi) kac-rice approx}

Recalling the notation from section \ref{sec:short arcs}
we now divide the domain of the integration, namely, the cube $S:=I^{2}=[0,L]^{2}$ into
small cubes $S_{ij}=I_{i}\times I_{j}$ of side $\delta_{0}$; some of the latter
will be designated as ``singular" and the rest as ``nonsingular".
Let $\epsilon_{1}>0$ be a small number that will be fixed till the end
(e.g. $\epsilon_{1}=\frac{1}{100}$ is sufficient).

\begin{definition}(Singular and nonsingular cubes and sets.)

\label{def:singular}

\begin{enumerate}
\item We call a point $(t_{1},t_{2})\in [0,L]^{2}$ {\em singular}
if either $|r (t_{1},t_{2})| > \epsilon_{1}$ or
$|r_{1}(t_{1},t_{2})| > \epsilon_{1}\cdot \sqrt{m}$ or
$|r_{2}(t_{1},t_{2})| > \epsilon_{1}\cdot \sqrt{m}$ or
$|r_{12}(t_{1},t_{2})| > \epsilon_{1}\cdot m$.

\item Let
$$S_{ij} = I_{i}\times I_{j} = [i\delta_{0}, (i+1)\delta_{0}] \times [j\delta_{0}, (j+1)\delta_{0}]$$ be a cube in $[0,L]^{2}$.
We say that $S_{ij}$ is a singular cube if it contains a
singular point.

\item The union of all the singular cubes is the singular set
$$B=B_{\eigen} = \bigcup\limits_{S_{ij}\text{ singular}}S_{ij}.$$

\end{enumerate}

\end{definition}

Note that outside the singular set $\Sigma(t_{1},t_{2})$ is
nonsingular (provided that $\epsilon_{1}$ is chosen sufficiently
small) by \eqref{eq: mu def}, \eqref{eq:rho def} and
\eqref{eq:Delta(rho) def}; we are thereupon allowed to apply the
Kac-Rice formula on $S\setminus B$; in particular for all $i,j$ with
$S_{i,j}\cap Int (B)=\emptyset$ (this implies $i\ne j$):
\begin{equation*}
\E[\mathcal{Z}_{i} \mathcal{Z}_{j}] =
\int\limits_{S_{ij}}K_{2}(t_{1},t_{2})dt_{1}dt_{2}.
\end{equation*}

We plan to approximate the $2$-point correlation function as
the corresponding sum of powers of $r$ and its derivatives; by
Proposition \ref{prop:K2m expansion r moments} we are allowed to do
so unless $r$ is big, and we will bound the contribution of the
domain where it is.

\begin{lemma}
\label{lem:on B either r or der big} If $S_{ij}\subseteq
B$ is singular, then for all $(t_{1},t_{2})\in S_{ij}$ either
$r(t_{1},t_{2})>\epsilon_{1}/2$ or the analogous statement holds for
one of the derivatives in the definition of singular point
(Definition \ref{def:singular} (i)).
\end{lemma}

\begin{proof}
The statement for $c_{0}$ sufficiently small follows from the fact
that the scaled covariance function $r_{F}(y/\sqrt{\eigen})$ of the
ambient field $F$ and its derivatives are Lipschitz with a universal
constant (independent of $\eigen$) (as it is easy to check, first
for the individual function $x\mapsto \cos(2\pi \langle \mu,x
\rangle)$, and then for their average), and thus the same holds for
$r$.
\end{proof}

\begin{lemma}
\label{lem:num sing cubes = O(r^4*m)} The total area of the
singular set is
\begin{equation*}
\meas(B)=O\left(N_{\eigen}^{-3/2} \right).
\end{equation*}
\end{lemma}

\begin{proof}
We apply the Chebyshev-Markov inequality on the measure of $B$.
Lemma \ref{lem:on B either r or der big} shows that it is bounded
from above by
\begin{equation*}
\operatorname{meas}(B) \ll \int\limits_{0}^{L}\bigg(
r(t_{1},t_{2})^{4}+\frac{1}{m^{2}}r_1(t_{1},t_{2})^{4}
+\frac{1}{m^{2}}r_2(t_{1},t_{2})^{4} +
\frac{1}{m^{4}}r_{12}(t_{1},t_{2})^{4} \bigg)dt_{1}dt_{2},
\end{equation*}
which is small by Lemma \ref{lem:4th moment r and der along gamma}
(which is independent of the arguments of the present section).
\end{proof}

Recall that $B$ consists of cubes of side length $\delta\asymp \frac{1}{\sqrt{\eigen}}$ 
(see \eqref{eq:delta<<>>1/sqrt(m)}).
Corollary \ref{cor:Cov(Zi,Zj)=O(1)} implies that the number of singular cubes is $\ll \frac{\eigen}{N_{\eigen}^{3/2}}$
and, teamed with Lemma \ref{lem:num sing cubes = O(r^4*m)}, yields the following estimate on the total
contribution of the singular domain $B$.
\begin{corollary}
\label{cor:tot contr sing} The total contribution of the singular
set is:
\begin{equation*}
\left| \sum\limits_{S_{ij}\text{ singular}}\Cov\left(\mathcal{Z}_{i},
\mathcal{Z}_{j}\right)\right| = O(m\cdot N_{m}^{-3/2}).
\end{equation*}
\end{corollary}

\begin{proof}[Proof of Proposition \ref{prop:var(Z) kac-rice approx}]

Consider the equality \eqref{eq:tot var = sum covar} and apply
Kac-Rice on every nonsingular cube (i.e. use \eqref{eq:CovIiIj disjoint}
for those $I_{i}$ and $I_{j}$ such that $S_{ij}$ is not lying in $B$, bearing in mind that 
for all $(t_{1},t_{2})\in S_{ij}$, $\Sigma(t_{1},t_{2})$ is nonsingular). We then
obtain
\begin{equation*}
\begin{split}
\var(\mathcal{Z}) &= \int\limits_{S\setminus B}(K_{2}(t_{1},t_{2})-K_{1}(t_{1})K_{1}(t_{2}))dt_{1}dt_{2} + \sum\limits_{S_{ij}
\text{ singular}} Cov(\mathcal{Z}_{i}, \mathcal{Z}_{j})
\\&= \int\limits_{S\setminus B}(K_{2}(t_{1},t_{2})-K_{1}(t_{1})K_{1}(t_{2}))dt_{1}dt_{2} +
O(\eigen\cdot N_{m}^{-3/2}),
\end{split}
\end{equation*}
by Corollary \ref{cor:tot contr sing}. We finally use the expansion in Proposition
\ref{prop:K2m expansion r moments} for $K_{2}$ valid outside of $B$
(the latter of the two equivalent forms),
and use Lemma \ref{lem:4th moment r and der along gamma} again for bounding
the contribution of the error term in
\eqref{eq:K2m expansion r moments}, together with the everywhere boundedness of
the integrand on the rhs of \eqref{eq:var(Z) kac-rice approx} to conclude the proof.
\end{proof}

\subsection{Proof of Proposition \ref{prop:K2 diag =O(m)}}

\label{sec:K2 diag =O(m)}

\begin{proof}
From Lemma ~\ref{eq:K2 explicit}, since $1 \le G\le \frac{\pi}{2}$,
\begin{equation*}
K_{2}(t_{1},t_{2}) \ll \frac{1}{(1-r^{2})^{3/2}} \cdot \mu \ll
\frac{1}{(1-r)^{3/2}}\sqrt{\alpha(1-r^{2})-r_{1}^{2}}\cdot
\sqrt{\alpha(1-r^{2})-r_{2}^{2}}.
\end{equation*}
Note that
\begin{equation}
\label{eq:K2 bnd 1/sqrt(1-r)}
\begin{split}
&\frac{1}{(1-r^{2})^{3/2}}\sqrt{\alpha(1-r^{2})-r_{1}^{2}}\cdot
\sqrt{\alpha(1-r^{2})-r_{2}^{2}}
\\&= \frac{\alpha}{\sqrt{1-r^{2}}} \sqrt{1-\frac{r_{1}^{2}}{\alpha(1-r^{2})}}\cdot
\sqrt{1-\frac{r_{2}^{2}}{\alpha(1-r^{2})}} \\&\ll
\frac{\alpha}{\sqrt{1-r}}
\sqrt{1-\frac{r_{1}^{2}}{\alpha(1-r^{2})}}\cdot
\sqrt{1-\frac{r_{2}^{2}}{\alpha(1-r^{2})}} \le
\frac{\alpha}{\sqrt{1-r}}.
\end{split}
\end{equation}

The diagonal cube $S=S_{ij}$ contains a point of the form
$(t_{1},t_{1})$. We may Taylor expand the integrand
$K_{2}(t_{1},t_{2})$ for $(t_{1},t_{2})\in S$ about $(t_{1},t_{2})$
as a function of $t_{2}$, $t_{1}$ {\em fixed}, and assuming WLOG
$t_{2}>t_{1}$.

To expand $r$ we differentiate and evaluate the derivatives at the
diagonal $t_{2}=t_{1}$: The first derivative $r_2=\partial
r/\partial t_2$ is
\begin{equation*}
r_{2} = -\nabla r_{F_{m}}(\gamma(t_{1})-\gamma(t_{2})) \cdot
\dot{\gamma}(t_{2}),
\end{equation*}
and on the diagonal 
\begin{equation}
\label{eq:r22(t1,t2)}
r_{2}(t,t)=0. 
\end{equation}
The second derivative $r_{22} =
\partial^2 r/\partial t_2^2$ is
\begin{equation*}
r_{22} = \dot{\gamma}(t_{2})^{t}\cdot
H_{r_{F_{m}}}(\gamma(t_{1})-\gamma(t_{2})) \cdot \dot{\gamma}(t_{2})
- \nabla r_{F_{m}}(\gamma(t_{1})-\gamma(t_{2})) \cdot
\ddot{\gamma}(t_{2}),
\end{equation*}
on the diagonal $r_{22}(t,t)=-\alpha$. The third derivative is
\begin{equation}
\label{eq:r222(t1,t2)}
\begin{split}
r_{222} &= \frac{\partial}{\partial t_{2}}\left[
\dot{\gamma}(t_{2})^{t}\cdot
H_{r_{F_{m}}}(\gamma(t_{1})-\gamma(t_{2})) \cdot \dot{\gamma}(t_{2})
\right]\\&+ \dot{\gamma}(t_{2})\cdot
H_{r_{F_{m}}}(\gamma(t_{1})-\gamma(t_{2})) \cdot
\ddot{\gamma}(t_{2}) -\nabla r_{F_{m}}(\gamma(t_{1})-\gamma(t_{2}))
\cdot \dddot{\gamma}(t_{2})
\\&=  \dot{\gamma}(t_{2})^{t}\cdot \frac{\partial}{\partial t_{2}}\left[H_{r_{F_{m}}}(\gamma(t_{1})-\gamma(t_{2}))\right] \cdot \dot{\gamma}(t_{2}) \\&+ 3\dot{\gamma}(t_{2})^{t}\cdot H_{r_{F_{m}}}(\gamma(t_{1})-\gamma(t_{2})) \cdot \ddot{\gamma}(t_{2})
-\nabla r_{F_{m}}^{t}(\gamma(t_{1})-\gamma(t_{2})) \cdot
\dddot{\gamma}(t_{2}),
\end{split}
\end{equation}
and on the diagonal
\begin{equation}
\label{eq:r222 diag=0}
r_{222}(t_{1},t_{1}) = -3\alpha\dot{\gamma}(t_{2})^{t}\cdot
\ddot{\gamma}(t_{2}) = 0,
\end{equation}
since the acceleration is always orthogonal to the velocity ($t$ is
the arc-length parameter). Moreover, the Hessian satisfies $H\ll m$
and $\partial H/\partial t_2\ll m^{3/2}$ everywhere, so that we have
\begin{equation*}
r_{222}(t_{1},t_{2})= O(m^{3/2})
\end{equation*}
everywhere.

The expansion of $r(t_{1},t_{2})$ around the diagonal $t_{2}=t_{1}$,
valid for $0 < t_{2}-t_{1} \le \frac{c_{0}}{\sqrt{\eigen}}$ with
$c_{0}$ sufficiently small, is
\begin{equation*}
r= 1-\frac{\alpha}{2}(t_{2}-t_{1})^{2}+O(m^{3/2}(t_{2}-t_{1})^{3}),
\end{equation*}
and
\begin{equation}
\label{eq:(1-r^2) taylor diag}
\begin{split}
1-r^{2} &= (1-r)(1+r) \\&= \left[ \frac{\alpha}{2}(t_{2}-t_{1})^{2}
+
O\left(m^{3/2}(t_{2}-t_{1})^{3}\right)\right]\left[2-\frac{\alpha}{2}(t_{2}-t_{1})^{2}
+ O\left(m^{3/2}(t_{2}-t_{1})^{3}\right)\right]
\\&= \alpha(t_{2}-t_{1})^{2}\left( 1+O(\sqrt{m}(t_{2}-t_{1}))\right),
\end{split}
\end{equation}
\begin{equation*}
r_{2}^{2} \approx r_{1}^{2} =
\alpha^{2}(t_{2}-t_{1})^{2}\left(1+O\left(m^{1/2}(t_{2}-t_{1})\right)
\right),
\end{equation*}
thus
\begin{equation*}
\frac{r_{1}^{2}}{\alpha(1-r^{2})} =
1+O\left(m^{1/2}(t_{2}-t_{1})\right),
\end{equation*}
and hence
\begin{equation*}
0 \le 1- \frac{r_{1}^{2}}{\alpha(1-r^{2})} =
O\left(m^{1/2}(t_{2}-t_{1})\right),
\end{equation*}
and the same estimate holds for $$1-
\frac{r_{2}^{2}}{\alpha(1-r^{2})}.$$

Consolidating all the estimates we conclude that \eqref{eq:K2 bnd
1/sqrt(1-r)} is uniformly bounded by
\begin{multline*}
\frac{\alpha}{\sqrt{1-r}}
\sqrt{1-\frac{r_{1}^{2}}{\alpha(1-r^{2})}}\cdot
\sqrt{1-\frac{r_{2}^{2}}{\alpha(1-r^{2})}} \\
\ll \frac{\alpha}{m^{1/2}(t_{2}-t_{1})}\cdot O(m^{1/2}(t_{2}-t_{1}))
= O(\eigen),
\end{multline*}
recalling that $\alpha=2\pi^{2}\eigen$.
\end{proof}

\section{Asymptotics for the second moments of the covariance function and its derivatives}

\label{sec:r,r1,r2,r12 2nd mom}

Recall that $r$ is the covariance function restricted to the curve
$\curve$:
\begin{equation}
r(t_1,t_2 )= r(\gamma(t_1),\gamma(t_2))
\end{equation}

\begin{proposition}\label{second mom of r}
If $\curve \subset \TT^2$ is a (smooth) curve with nowhere
vanishing curvature, then for all $\epsilon>0$
\begin{equation}\label{2moment of r}
\int_\curve \int_\curve r^2 =  \int_0^L\int_0^L  r(t_1,t_2)^2 dt_1
dt_2 = \frac {L^2}{N_\eigen} + O\left(\frac {
1}{N_\eigen^{2-\epsilon}}\right)
\end{equation}
\begin{equation}
\int_\curve \int_\curve \left|\frac 1{ 2\pi\sqrt{\eigen}} \frac
{\partial r}{\partial t_1} \right|^2  =  \frac {L^2}{2N_\eigen} +
O\left(\frac { 1}{N_\eigen^{2-\epsilon}}\right)
\end{equation}
and
\begin{equation}
\int_\curve \int_\curve  \left|\frac 1{4\pi^2 \eigen}
\frac{\partial^2 r}{\partial t_1 \partial t_2}\right|^2 =   \frac
{B_\curve(\vE)}{N_\eigen} + O\left(\frac {
1}{N_\eigen^{2-\epsilon}}\right)
\end{equation}
where
\begin{equation} \label{def of B}
B_\curve(\vE):= \int_{\curve}\int_{ \curve} \frac 1N_\eigen
\sum_{\mu \in \vE} \left\langle \frac{\mu}{|\mu|}, \dot \gamma(t_1)
\right\rangle^2 \cdot  \left\langle \frac{\mu}{|\mu|}, \dot
\gamma(t_2) \right\rangle^2 dt_1 dt_2.
\end{equation}
\end{proposition}

Before proceeding with the proof, we can conclude the proof of
Theorem \ref{thm:nodal var intr asympt}: Use Proposition
\ref{prop:var(Z) kac-rice approx} to write an approximate integral
formula for the nodal intersections number variance and substitute
the result of Proposition \ref{second mom of r} in place of the main
term of \eqref{prop:var(Z) kac-rice approx}. \qed


\subsection{Main terms}
Squaring out, we have (on isolating the diagonal pairs $\mu=\mu'$)
\begin{equation}
|r(t_1,t_2)|^2 = \frac 1N_\eigen+ \frac 1{N_\eigen^2} \sum_{\substack{\mu,\mu'\in \vE\\
\mu\neq \mu'}} e^{2\pi i\langle
\mu-\mu',\gamma(t_1)-\gamma(t_2)\rangle}
\end{equation}
and hence integrating we find
\begin{equation}
\iint |r(t_1,t_2)|^2 dt_1dt_2 = \frac {L^2}{N_\eigen} +\frac 1{N_\eigen^2}\sum_{\substack{\mu,\mu'\in \vE\\
\mu\neq \mu'}}  \left| \int_0^L e^{2\pi i\langle
\mu-\mu',\gamma(t)\rangle} dt \right|^2.
\end{equation}

For the second moment of the derivative $r_1= \partial/\partial t_1$
we compute
\begin{equation}
\frac 1{2\pi i \sqrt{\eigen}} \frac {\partial r}{\partial t_1}(t_1,t_2) =
\frac 1{N_\eigen} \sum_\mu  \left\langle \frac{\mu}{|\mu|}, \dot
\gamma(t_1) \right\rangle e^{2\pi i\langle
\mu,\gamma(t_1)-\gamma(t_2)\rangle}
\end{equation}
and setting
\begin{equation}
A_{\mu,\mu'}(t) = \left\langle \frac{\mu}{|\mu|}, \dot \gamma(t)
\right\rangle \left\langle \frac{\mu'}{|\mu'|}, \dot \gamma(t)
\right\rangle
\end{equation}
we find
\begin{multline}\label{expand first derivative}
\iint \left|\frac 1{2\pi \sqrt{\eigen}} \frac {\partial r}{\partial
t_1}(t_1,t_2)\right|^2 dt_1dt_2 =   \frac 1{N_\eigen^2} \sum_\mu
\int_0^L A_{\mu,\mu}(t_1)dt_1 \int_0^L 1 dt_2
\\
+\frac 1{N_\eigen^2}\sum_{\substack{\mu,\mu'\in \vE\\
\mu\neq \mu'}}  \int_0^L A_{\mu,\mu'}(t_1)   e^{2\pi i\langle
\mu-\mu',\gamma(t_1)\rangle} dt_1
 \int_0^L e^{2\pi i\langle
\mu'-\mu,\gamma(t_2)\rangle} dt_2.
\end{multline}

Similarly,
\begin{multline}\label{expand second derivative}
\iint \left|\frac 1{4\pi^2\eigen} \frac{\partial^2 r}{\partial t_1
\partial t_2}(t_1,t_2)\right|^2 dt_1dt_2 = \frac 1{N_\eigen^2} \sum_\mu
\iint A_{\mu,\mu}(t_1)A_{\mu,\mu}(t_2)dt_1dt_2
\\
+\frac 1{N_\eigen^2}\sum_{\substack{\mu,\mu'\in \vE\\
\mu\neq \mu'}}  \left| \int_0^L A_{\mu,\mu'}(t) e^{2\pi i\langle
\mu-\mu',\gamma(t)\rangle} dt \right|^2.
\end{multline}

For $\partial r/\partial t_1$ we use (see \cite[Lemma 2,3]{RW}) that
for any $v\in \R^2$,
\begin{equation}
\frac 1N_\eigen \sum_{\mu\in \vE} \langle \mu, v \rangle^2 = \frac
\eigen 2 ||v||^2
\end{equation}
and applying it for $v=\dot \gamma(t)$ which has unit length we get
that
\begin{equation}
\frac 1N_\eigen \sum_\mu A_{\mu,\mu}(t) =\frac 12 ||\dot
\gamma(t)||^2 = \frac 12.
\end{equation}
Integrating over $t_1$ and $t_2$ shows that the diagonal term in
\eqref{expand first derivative} is $L^2/2N_\eigen$.

For $\partial^2 r/\partial t_1\partial t_2$ the diagonal term in
\eqref{expand second derivative} is
\begin{equation}
\frac 1N_\eigen \iint \frac 1N_\eigen \sum_\mu \left\langle
\frac{\mu}{|\mu|}, \dot \gamma(t_1) \right\rangle^2 \cdot
\left\langle \frac{\mu}{|\mu|}, \dot \gamma(t_2) \right\rangle^2
dt_1 dt_2 = \frac{B_\curve(\vE)} {N_\eigen}.
\end{equation}

\subsection{Off-diagonal terms}
To handle the off-diagonal terms $\mu\neq \mu'$, we need the
following consequence of van der Corput's lemma (see \cite{BRNI}):
For each $0\neq \xi \in \R^2$ define a phase function on the curve
$\curve$ by
\begin{equation}
\phi_\xi(t) = \left\langle \frac{\xi}{|\xi|},\gamma(t) \right\rangle.
\end{equation}
Let $A\in C^\infty(0,L)$ be a smooth amplitude and for $k$ real, set
\begin{equation}
I(k) = \int A(t)e^{i k \phi_\xi(t)}dt.
\end{equation}
\begin{lemma} \label{osc int curve}
Assume $\curve$ has nowhere vanishing curvature. Then for $|k|\geq
1$,
\begin{equation}
|I(k)|\ll \frac 1{|k|^{1/2}}\left\{||A||_\infty + ||A'||_1 \right\},
\end{equation}
the implied constant depending only on the curve $\curve$
(independent of $\xi$).
\end{lemma}

Applying Lemma~\ref{osc int curve}, we see that for
  $\mu\neq \mu'$,
\begin{equation}
\ \int_0^Le^{2\pi i\langle \mu-\mu',\gamma(t)\rangle} dt \ll_\curve
\frac 1{|\mu-\mu'|^{1/2}}.
\end{equation}
Moreover, $|A_{\mu,\mu'}|\leq 1$ and $|A_{\mu,\mu'}'|\leq 2
K_{\max}$ where $K_{\max}$ is the maximum value of the curvature on
$\curve$, because
\begin{equation}
\begin{split}
A_{\mu,\mu'}' &=   \left\langle \frac{\mu}{|\mu|}, \ddot \gamma(t)
\right\rangle \left\langle \frac{\mu'}{|\mu'|}, \dot \gamma(t)
\right\rangle +   \left\langle \frac{\mu}{|\mu|}, \dot \gamma(t)
\right\rangle
\left\langle \frac{\mu'}{|\mu'|}, \ddot \gamma(t) \right\rangle =\\
& =\kappa(t) \left( \left\langle  \frac{\mu}{|\mu|}, \nu(t)
\right\rangle \cdot \left\langle \frac{\mu'}{|\mu'|}, \dot \gamma(t)
\right\rangle +  \left\langle \frac{\mu}{|\mu|}, \dot \gamma(t)
\right\rangle \left\langle \frac{\mu'}{|\mu'|}, \nu(t) \right\rangle
\right),
\end{split}
\end{equation}
where $\ddot \gamma = \kappa \nu$ with $\kappa$ the curvature and
$\nu$ the unit normal to the curve. Therefore we likewise find
\begin{equation}
 \int_0^L A_{\mu,\mu'}(t) e^{2\pi i\langle
\mu-\mu',\gamma(t)\rangle} dt \ll_\curve \frac 1{|\mu-\mu'|^{1/2}}.
\end{equation}

Hence we find that
\begin{equation}
\iint |r(t_1,t_2)|^2 dt_1dt_2 = \frac {L^2}{N_\eigen} +O\left( \frac 1{N_\eigen^2} \sum_{\substack{\mu,\mu'\in \vE\\
\mu\neq \mu'}}  \frac 1{|\mu-\mu'|} \right)
\end{equation}
and for $j=1,2$
\begin{equation}
\iint \left|\frac 1{ 2\pi \sqrt{\eigen}} \frac{\partial r}{\partial t_j} (t_1,t_2)\right|^2 dt_1dt_2 = \frac {L^2}{2N_\eigen} +O\left( \frac 1{N_\eigen^2} \sum_{\substack{\mu,\mu'\in \vE\\
\mu\neq \mu'}}  \frac 1{|\mu-\mu'|} \right),
\end{equation}
and finally
\begin{equation}
\iint \left|\frac 1{4\pi^2 \eigen} \frac{\partial^2 r}{\partial t_1\partial t_2} (t_1,t_2)\right|^2 dt_1dt_2 = \frac{B_\curve(\vE)}{N_\eigen} +O\left( \frac 1{N_\eigen^2} \sum_{\substack{\mu,\mu'\in \vE\\
\mu\neq \mu'}}  \frac 1{|\mu-\mu'|} \right).
\end{equation}

Proposition~\ref{second mom of r} hence follows from
\begin{proposition}\label{prop:sum}
\begin{equation}
\sum_{\substack{\mu,\mu'\in \vE\\
\mu\neq \mu'}}  \frac 1{|\mu-\mu'|}  \ll  N_\eigen^\epsilon, \quad
\forall \epsilon>0.
\end{equation}
\end{proposition}


\subsection{A result of Mordell}

Denote by $\mathcal H$ the set of $h\leq H$ for which the system
\begin{equation}\label{system h}
|\mu|^2 = \eigen=|\mu'|^2,\; |\mu-\mu'|^2=2h
\end{equation}
has integer solutions, and by $A(\eigen,h)$ the number such
solutions.

We   give an arithmetic characterization of the set $\mathcal H$.
To do so, we will need a result of Mordell \cite{Mordell 1932} (see
also Niven \cite{Niven})
 on the representation of a binary quadratic form as a sum of two squares of integer linear forms.
\begin{theorem}[Mordell \cite{Mordell 1932}]
 Let $A,B,C\in \Z$. Assume that the integer binary quadratic form
$$F(x,y):=Ax^2+2Bxy +Cy^2$$
is  positive definite, i.e. that $A,C>0$ and $AC-B^2>0$.
Then we can represent
$$ F(x,y) = (ux+u'y)^2 + (vx+v'y)^2$$
with integer $u,v,u',v'$ if and only if
\begin{equation}\label{square disc cond}
AC-B^2 = \square \mbox{ is a perfect square},
\end{equation}
and
\begin{equation}\label{gcd is sum of 2 cond}
 \gcd(A,B,C) = \square+\square \mbox{  is a sum of two integer
squares}.
\end{equation}
\end{theorem}

Pall \cite{Pall} gives the exact number of solutions as
$r_2(\gcd(A,B,C))$ if $AC-B^2>0$, and $2r_2(\gcd(A,B,C))$ if
$AC-B^2=0$, where $r_2(n)$ is the number of representations of $n$
as a sum of two integer squares.

 Writing $\mu
= (u,v)$ and $\mu'=(u',v')$ we have
$$
(ux+u'y)^2 + (vx+v'y)^2 = |x\mu+y\mu'|^2
$$
so that we can interpret Mordell's theorem as saying that given
$A,B,C$ as above, there are integer vectors $\mu,\mu'\in \Z^2$
satisfying
\begin{equation}
|\mu|^2 = A,\quad \langle \mu,\mu' \rangle = B, \quad |\mu'|^2 = C
\end{equation}
if and only if \eqref{square disc cond} and \eqref{gcd is sum of 2
cond} hold.

A consequence is
\begin{corollary}\label{cor:Mordell}
Let $\eigen,h\in \Z$, $0<h<\eigen$. There are two integer vectors
$\mu,\mu'$ with $|\mu|^2 = \eigen = |\mu'|^2$ and $|\mu-\mu'|^2 =
2h$ if and only if
\begin{enumerate}
\item\label{cond1}
   $ h(2\eigen-h)=\square$ is a perfect square, and
\item \label{cond2}
   $  \gcd(\eigen,h)=\square + \square$ is a sum of two squares.
\end{enumerate}
In this case the number of solutions is
$A(\eigen,h)=r_2(\gcd(\eigen,h))\ll h^{o(1)}$.
\end{corollary}

\subsection{Proof of Proposition~\ref{prop:sum}}

Let $H=N_\eigen^4$. We separate the sum into that over "distant"
pairs $|\mu-\mu'|^2>H$ and "close" pairs $1\leq |\mu-\mu'|^2\leq H$.
For the sum over distant pairs, we crudely use
\begin{equation}\label{distant}
\sum_{\substack{\mu,\mu'\in \vE\\|\mu-  \mu'|^2>H}} \frac
1{|\mu-\mu'|} \ll \frac {N_\eigen^2}{\sqrt{H}}.
\end{equation}

To handle the sum over "close" pairs, we write
\begin{equation}
\sum_{ \substack{(\mu,\mu')\in \vE \times \vE\\0<|\mu-\mu'|^2 < H}}
\frac 1{|\mu-\mu'|} = \sum_{\substack{h\in \mathcal H\\ 0<h<H}}
\frac {A(\eigen,h)}{\sqrt{2h}}\ll H^\epsilon \sum_{\substack{h\in
\mathcal H\\ 0<h<H}} \frac {1}{\sqrt{ h}}.
\end{equation}
For $h\in \mathcal H$, we write $d=(h,\eigen)$, which is a sum of
two squares ($d=\Box+\Box$), $h=dh'$, $\eigen = d\eigen'$ with
$(\eigen',h')=1$. Then $h\in \mathcal H$ means $h(2\eigen-h)=\Box$
and so $h'(2\eigen'-h')=\Box$. Thus we find

\begin{equation}\label{eq:double sum}
\sum_{\substack{h\in \mathcal H\\ 0<h<H}} \frac 1{\sqrt{ h}}  =
\sum_{\substack{d\mid\eigen\\d=\Box+\Box\\d<H}}\frac 1{\sqrt{d}}
\sum_{\substack{ h'(2\eigen'-h')=\Box\\(h',\eigen')=1\\ h'<H/d}}
\frac 1{\sqrt{h'}}.
\end{equation}

We claim that the inner sum over $h'$ is $O(1)$. To see this, use
$1/\sqrt{h'}\leq 1$ and separate into cases according to $h'$ being
odd or even. If $h'$ is odd and $(h',\eigen')=1$, then the condition
$h'(2\eigen'-h')=\Box$ implies $h'=\Box$ and $2\eigen'-h'=\Box$,
that is $h'=u^2$ and $2\eigen-h'=v^2$ with $v>0$, $0<u<\sqrt{H/d}$.
If $h'$ is even, the the condition $h'(2\eigen'-h')=\Box$ and
$(h',\eigen')=1$ implies $(h'/2,\eigen'-h'/2)=1$ and $h'/2=\Box$,
$\eigen'-h'/2-\Box$ so that $h'/2=u^2$, $\eigen'-h'/2=v^2$ with
$v>0$, $0<u<\sqrt{H/d}$. Summarizing, we get lattice points on the
circle $u^2+v^2=2\eigen'$ or $u^2+v^2=\eigen'$ depending on the
parity of $h'$, with $0<u<\sqrt{H/d}$, $v>0$. These conditions puts
these lattice points on a "short" arc on the circle, since $H\ll
\eigen^{o(1)}$. Recall  Jarnik's theorem \cite{Jarnik}, which states
that on an arc of size $<R^{1/3}$ on a circle of radius $R$ there
can be at most two lattice points. Hence there are at most two such
lattice points in each of the two cases, and thus the number of
participating $h'$ is at most $4$. This proves that the inner sum in
\eqref{eq:double sum} is bounded.

We conclude that
\begin{equation}
\sum_{\substack{h\in \mathcal H\\ 0<h<H}} \frac 1{\sqrt{ h}} \ll
\sum_{\substack{d\mid\eigen\\d=\Box+\Box\\d<H}}\frac 1{\sqrt{d}}.
\end{equation}
Below in Lemma~\ref{lem:divisorsum} we will show that this sum is
bounded by $O(N_\eigen^\epsilon)$. This will show that the
contribution of close pairs is $O(N_\eigen^\epsilon)$. Combining
with the bound \eqref{distant} on distant pairs we get
\begin{equation}
\sum_{\substack{\mu,\mu'\in \vE\\
\mu\neq \mu'}}  \frac 1{|\mu-\mu'|}  \ll \frac{N_m^2}{\sqrt{H}}
+N_\eigen^\epsilon \ll N_\eigen^\epsilon
\end{equation}
on recalling that $H=N_\eigen^4$. This will conclude the proof of
Proposition~\ref{prop:sum}, once we prove:
\begin{lemma}\label{lem:divisorsum}
Suppose that $H=N_\eigen^\alpha$ for some $\alpha>0$. Then
$$ \sum_{\substack{d\mid\eigen\\d=\Box+\Box\\d<H }}\frac 1{\sqrt{d}}
\ll _\epsilon N_\eigen^\epsilon.
$$
\end{lemma}
\begin{proof}
Write $\eigen = \eigen_1^2\eigen_1$ where $\eigen_1 =
2^r\prod_{q_k=3\bmod 4} q_k^{b_k}$ is a product of powers of primes
$q_k=3\bmod 4$ and possibly a power of $2$, and
$\eigen_2=2^{c}\prod_j p_j^{a_j}$ is a product of powers of primes
$p_j=1\bmod 4$, possibly times $2$ ($c=0,1$). Then
$$
N_\eigen =\prod_j(a_j+1).
$$
Likewise we write $d=d_1^2d_2$ in the same fashion, so that  $d\mid
\eigen$ is equivalent to $d_1\mid \eigen_1$ and $d_2\mid \eigen_2$.

The sum over $d$'s is bounded by
$$
\sum_{\substack{d\mid\eigen\\d=\Box+\Box\\d<H }} \frac 1{\sqrt{d}}
\ll \sum_{\substack{d_1 \mid \eigen_1 \\ d_1<\sqrt{H}}}  \frac
1{d_1} \sum_{d_2\mid \eigen_2} \frac 1{\sqrt{d_2}} \ll \log H
\sum_{d_2\mid \eigen_2} \frac 1{\sqrt{d_2}} \ll \log N_\eigen
\sum_{d_2\mid \eigen_2} \frac 1{\sqrt{d_2}},
$$
where in the sum over $d_2$  we have dropped the condition $d<H$.

It now suffices to show that for all $\epsilon>0$, there is some
$C(\epsilon)>0$ so that
$$
\sum_{d_2\mid \eigen_2} \frac 1{\sqrt{d_2}} \leq C(\epsilon)
N_\eigen^\epsilon.
$$

Ignoring the possible factor of $2$,
$$ \sum_{d_2\mid \eigen_2}
\frac 1{\sqrt{d_2}} \ll \prod_j \left(1+\frac 1{\sqrt{p_j}} +\dots +
\frac 1{p_j^{(a_j+1)/2}} \right) \leq \prod_j \frac 1{1-\frac
1{\sqrt{p_j}}}.
$$
Recalling that $N_\eigen = \prod_j(a_j+1)\geq \prod_j 2$ we find
\begin{equation}
\begin{split}
\frac{1}{N_\eigen^\epsilon}  \sum_{d_2\mid \eigen_2} \frac
1{\sqrt{d_2}}   & \ll\prod_{\substack{p_j\mid \eigen\\p_j=1\bmod
4}} \frac 1{(1-\frac 1{\sqrt{p_j}}) 2^\epsilon} \\
&\leq \prod_{*}\frac 1{(1-\frac 1{\sqrt{p}})2^\epsilon}=:C(\epsilon),
\end{split}
\end{equation} where in the last line, the product is over all
primes satisfying $(1-\frac 1{\sqrt{p}})2^\epsilon<1$. This gives
$\sum_{d_2\mid \eigen_2}   1/ \sqrt{d_2} \leq C(\epsilon)
N_\eigen^\epsilon $ as claimed.
\end{proof}

\section{Bounds for the higher moments of $r$ and its derivatives}
\label{sec:4th moment r and der along gamma}

\begin{lemma}
\label{lem:4th moment r and der along gamma} We have the following
estimates on the $4$th moments of the covariance function and its
various derivatives along a (smooth) reference curve $\gamma$ with
nowhere vanishing curvature:
\begin{equation}
\label{eq:4th moment r along gamma}
\iint\limits_{[0,L]^{2}}r(t_{1},t_{2})^{4}dt_{1}dt_{2}
 = O\left( \frac{1}{N_{m}^{3/2} }\right),
\end{equation}
\begin{equation*}
\frac{1}{m^{2}}\iint\limits_{[0,L]^{2}}r_1(t_{1},t_{2})^{4}dt_{1}dt_{2}
 = O\left( \frac{1}{N_{m}^{3/2}} \right),
\end{equation*}
\begin{equation*}
\frac{1}{m^{2}}\iint\limits_{[0,L]^{2}}r_2(t_{1},t_{2})^{4}dt_{1}dt_{2}
 = O\left( \frac{1}{N_{m}^{3/2}} \right),
\end{equation*}
\begin{equation*}
\frac{1}{m^{4}}\iint\limits_{[0,L]^{2}}r_{12}(t_{1},t_{2})^{4}dt_{1}dt_{2}
 = O\left( \frac{1}{N_{m}^{3/2}} \right).
\end{equation*}

\end{lemma}

\begin{proof}
Abbreviating $e(z):=e^{2\pi i z}$, we have
\begin{multline*}
\iint\limits_{[0,L]^{2}}r(t_{1},t_{2})^{4}dt_{1}dt_{2} \\=
\frac{1}{N_{m}^{4}} \sum_{\mu_1,\dots,\mu_4\in \vE}\iint_{[0,L]^{2}}
e(\left\langle \mu_{1}+\mu_{2}+\mu_{3}+\mu_{4},
\gamma(t_{1})-\gamma(t_{2})\right\rangle)dt_{1}dt_{2}
\\= \frac{1}{N_{m}^{4}}
\sum_{\mu_1,\dots,\mu_4\in \vE}
\left|I_{1}(\mu_{1},\mu_{2},\mu_{3},\mu_{4})\right| ^{2}
\end{multline*}
with
\begin{equation}
\label{eq:I1 def} I_{1}(\mu_{1},\mu_{2},\mu_{3},\mu_{4}) =
\int\limits_{[0,L]} e\left(\left\langle
\mu_{1}+\mu_{2}+\mu_{3}+\mu_{4}, \gamma(t)\right\rangle \right)dt.
\end{equation}
Now by Lemma \ref{osc int curve}, for
$\mu_{1}+\mu_{2}+\mu_{3}+\mu_{4}\ne 0$ we have the estimate
\begin{equation}
\label{eq:I1<<1/|mu1+...mu4|^1/2}
|I_{1}(\mu_{1},\mu_{2},\mu_{3},\mu_{4})| \ll
\frac{1}{|\mu_{1}+\mu_{2}+\mu_{3}+\mu_{4}|^{1/2}}.
\end{equation}
Hence
\begin{equation*}
\begin{split}
&\iint_{[0,L]^{2}}r(t_{1},t_{2})^{4}dt_{1}dt_{2} \ll
\frac{1}{N_{m}^{2}} + \frac{1}{N_{m}^{4}} \sum_{\substack {\mu_1,\dots,\mu_4\in \vE\\
\mu_{1}+\mu_{2}+\mu_{3}+\mu_{4}\ne
0}}\frac{1}{\|\mu_{1}+\mu_{2}+\mu_{3}+\mu_{4}\|},
\end{split}
\end{equation*}
since for given $\mu_{1}, \mu_{2}\in\Ec$ with $\mu_{1}\ne -\mu_{2}$
there exist (precisely) $2$ choices for $\mu_{3},\mu_{4}\in\Ec$ so
that $$\mu_{1}+\mu_{2}+\mu_{3}+\mu_{4}=0,$$ by an elementary
argument due to Zygmund ~\cite{Zygmund}. The estimate \eqref{eq:4th
moment r along gamma} now follows from Lemma \ref{lem:sum
1/mu1+...mu4}.

For the derivative $r_1$ we have:
\begin{equation}
\label{eq:r1 4th moment << sum mui}
\begin{split}
&\iint_{[0,L]^{2}}r_1(t_{1},t_{2})^{4}dt_{1}dt_{2}=
\frac{(2\pi)^{4}}{N_{m}^{4}}\sum_{\mu_1,\dots,\mu_4\in
\vE}I_{2}(\mu_{1},\mu_{2},\mu_{3},\mu_{4})\cdot
\overline{I_{1}(\mu_{1},\mu_{2},\mu_{3},\mu_{4})},
\end{split}
\end{equation}
where $I_{1}$ was defined in \eqref{eq:I1 def}, and
\begin{equation}
\label{eq:I2 def} I_{2}(\mu_{1},\mu_{2},\mu_{3},\mu_{4}) = \int_0^L
e(\left\langle \mu_{1}+\mu_{2}+\mu_{3}+\mu_{4},
\gamma(t)\right\rangle)
\mu_{1}^{t}\dot{\gamma}(t)\mu_{2}^{t}\dot{\gamma}(t)
\mu_{3}^{t}\dot{\gamma}(t)\mu_{4}^{t}\dot{\gamma}(t)dt.
\end{equation}
We invoke Lemma \ref{osc int curve} again to yield the bound
\begin{equation}
\label{eq:I1<<m^{2}/|mu1+...mu4|^1/2} |I_{2}| \ll m^{2} \cdot
\frac{1}{|\mu_{1}+\mu_{2}+\mu_{3}+\mu_{4}|^{1/2}},
\end{equation}
so that combined with the estimate \eqref{eq:I1<<1/|mu1+...mu4|^1/2}
and \eqref{eq:r1 4th moment << sum mui} it implies
\begin{equation*}
\frac{1}{m^{2}}\iint\limits_{[0,L]^{2}}r_1(t_{1},t_{2})^{4}dt_{1}dt_{2}
\ll \frac{1}{N_{m}^{2}} + \frac{1}{N_{m}^{4}} \sum\limits_{\substack {\mu_1,\dots,\mu_4\in \vE\\
\mu_{1}+\mu_{2}+\mu_{3}+\mu_{4}\ne
0}}\frac{1}{\|\mu_{1}+\mu_{2}+\mu_{3}+\mu_{4}\|},
\end{equation*}
yielding the statement of the present lemma in this case as before,
via Lemma \ref{lem:sum 1/mu1+...mu4}. The argument for $r_{2}$ is
identical.

For the second mixed derivative $r_{12}$ we have:
\begin{equation*}
r_{12} (t_{1},t_{2}) = -\frac{(2\pi)^{2}}{N_{m}}\sum_{\mu\in \vE}
\mu^{t}\dot{\gamma}(t_{1}) \mu^{t}\dot{\gamma}(t_{2})
e\left(\left\langle \mu_{1}+\mu_{2}+\mu_{3}+\mu_{4},
\gamma(t_{1})-\gamma(t_{2})\right\rangle \right),
\end{equation*}
and
\begin{multline*}
r_{12} (t_{1},t_{2})^{4}=\\
\frac{(2\pi)^{8}}{N_{m}^{4}}\sum_{\mu_1,\dots,\mu_4\in \vE}
\prod_{j=1}^4 \langle \mu_j,\dot\gamma(t_1)\rangle\cdot \langle
\mu_j,\dot\gamma(t_2)\rangle
\times  e\left(\left\langle \mu_{1}+\mu_{2}+\mu_{3}+\mu_{4},
\gamma(t_{1})-\gamma(t_{2})\right\rangle \right)
\end{multline*}
so that by separation of variables and upon recalling \eqref{eq:I2
def}, we have
\begin{equation*}
\iint\limits_{[0,L]^{2}}r_{12}(t_{1},t_{2})^{4}dt_{1}dt_{2}=
\frac{(2\pi)^{8}}{N_{m}^{4}}\sum_{\mu_1,\dots,\mu_4\in \vE}
|I_{2}(\mu_{1},\mu_{2},\mu_{3},\mu_{4})|^{2}\;,
\end{equation*}
and invoking \eqref{eq:I1<<m^{2}/|mu1+...mu4|^1/2} (valid for
$\mu_{1}+\mu_{2}+\mu_{3}+\mu_{4}\ne 0$), we finally have
\begin{equation*}
\begin{split}
\frac{1}{m^{4}}\iint\limits_{[0,L]^{2}}r_{12}(t_{1},t_{2})^{4}dt_{1}dt_{2}
&\ll \frac{1}{N_{m}^{2}} + \frac{1}{N_{m}^{4}} \sum\limits_{\substack {\mu_1,\dots,\mu_4\in \vE\\
\mu_{1}+\mu_{2}+\mu_{3}+\mu_{4}\ne
0}}\frac{1}{\|\mu_{1}+\mu_{2}+\mu_{3}+\mu_{4}\|} \\&=
O\left(\frac{1}{N_{m}^{3/2}}\right),
\end{split}
\end{equation*}
by Lemma \ref{lem:sum 1/mu1+...mu4}.
\end{proof}

\begin{lemma}
\label{lem:sum 1/mu1+...mu4} We have the following bound
\begin{equation}
\label{eq:sum 1/mu1+...mu4}
\sum_{\substack{\mu_1,\dots,\mu_4\in \vE\\
\mu_{1}+\mu_{2}+\mu_{3}+\mu_{4}\ne 0}}
\frac{1}{\|\mu_{1}+\mu_{2}+\mu_{3}+\mu_{4}\|} =
O\left(N_{m}^{5/2}\right).
\end{equation}
\end{lemma}

\begin{proof}

Let us denote $v=\mu_{1}+\mu_{2}+\mu_{3}+\mu_{4}$. We choose a big
parameter $A>0$ and split the summation into $3$ ranges:
\begin{enumerate}

\item $\|v\| \le A$.

We invoke Zygmund's elementary observation ~\cite{Zygmund} again to
deduce that, given $\mu_{1}$ and $\mu_{2}$ and $v$ such that
$$\mu_{1}+\mu_{2}\ne v,$$ there are (at most) two choices for
$\mu_{3},\mu_{4}\in \Ec$ that solve
$$\mu_{1}+\mu_{2}+\mu_{3}+\mu_{4}=v .$$ Therefore we may bound the contribution to the
sum \eqref{eq:sum 1/mu1+...mu4} of this range as
\begin{equation}
\label{eq:sum mu1-4 v<=A}
\begin{split}
&\sum\limits_{\substack {\mu_1,\dots,\mu_4\in \vE\\
\|\mu_{1}+\mu_{2}+\mu_{3}+\mu_{4}\| \le A}}
\frac{1}{\|\mu_{1}+\mu_{2}+\mu_{3}+\mu_{4}\|} \\&\le
N_{m}^{2}\cdot\sum\limits_{\|v\|\le A} \frac{1}{\|v\|} \ll
N_{m}^{2}\int\limits_{1\le |x\|\le A} \frac{dx}{\|x\|} =
N_{m}^{2}\int\limits_{1}^{A}dt =A \cdot N_{m}^2,
\end{split}
\end{equation}
by comparing the sum $\sum\limits_{\|v\|\le A}\frac{1}{\|v\|}$ to
the integral $\int\limits_{1\le |x\|\le A} \frac{dx}{\|x\|}.$

\item $A \le \|v\| \le N_{m}^{3/2}$.

We claim that given $\mu_{1},\mu_{2},\mu_{3}$ there exist at most
$2$ lattice points $\mu_{4}$ that lie in the relevant range so that
$\|\mu_{1}+\mu_{2}+\mu_{3}+\mu_{4}\| \le N_\eigen^{3/2}$. Once
established the above, the contribution of this range is, bounding
the summands point-wise,
\begin{equation}
\label{eq:sum mu1-4 A<=v<N^3/2}
\sum\limits_{\substack {\mu_1,\dots,\mu_4\in \vE\\
A \le \|\mu_{1}+\mu_{2}+\mu_{3}+\mu_{4}\| \le N_{m}^{3/2}}}
\frac{1}{\|\mu_{1}+\mu_{2}+\mu_{3}+\mu_{4}\|} \le \frac{1}{A}\cdot
N_{m}^{3}.
\end{equation}

To see that indeed, given $\mu_{1},\ldots \mu_{3}$ there are at most
two vectors $\mu_{4}$ that return us to the relevant range, we
consider the geometric picture. Let $\mu_{1},\mu_{2},\mu_{3}$ be
fixed, define $w= \mu_{1}+\mu_{2}+\mu_{3}$ and suppose that there
exists a vector $\mu_{4}$ so that
$v=\mu_{1}+\mu_{2}+\mu_{3}+\mu_{4}$ satisfies $N_{m}/\log{N_{m}} \le
\|v\| \le N_{m}^{3/2}$ indeed. By the triangle inequality, the
vector $w$ satisfies $$\sqrt{m}-N_{m}^{3/2} \le \|w\| \le
N_{m}^{3/2}+ \sqrt{m};$$ adding the vector $\mu_{4}$ translates it
to a circle of a small radius $N_{m}\log{N_{m}}$ around the origin,
which means that $\mu_4$ has to be on a circular arc of angle
$\alpha$ of the order at most
$$\alpha \sim \sin\alpha\le \frac{N_{m}^{3/2}}{\sqrt{m}-N_{m}^{3/2}},$$
with arc length $\le \sqrt{m}\frac{N_{m}^{3/2}}{\sqrt{m}-N_{m}^{3/2}}$,
which is much smaller than $m^{1/3}$, so by Jarnik
there exists at most two such lattice points, as claimed.

\item $\|v\| \ge N_{m}^{3/2}$.

Here it is sufficient to bound the summands in \eqref{eq:sum
1/mu1+...mu4} pointwise; since the total number of summands is
$N_{m}^{4}$ the sum is bounded as
\begin{equation}
\label{eq:sum mu1-4 v>=N^3/2}
\sum_{\substack{\mu_1,\dots,\mu_4\in \vE\\
\|\mu_{1}+\mu_{2}+\mu_{3}+\mu_{4}\| \ge N_{m}^{3/2}}}
\frac{1}{\|\mu_{1}+\mu_{2}+\mu_{3}+\mu_{4}\|} \le
\frac{1}{N_{m}^{3/2}} \cdot N_{m}^4 = N_{m}^{5/2}.
\end{equation}

\end{enumerate}

Consolidating \eqref{eq:sum mu1-4 v<=A}, \eqref{eq:sum mu1-4
A<=v<N^3/2} and \eqref{eq:sum mu1-4 v>=N^3/2} we find that the sum
\eqref{eq:sum 1/mu1+...mu4} is bounded by
\begin{equation*}
A\cdot N_{m}^{2} + \frac{1}{A}\cdot N_{m}^{3}   +N_{\eigen}^{5/2},
\end{equation*}
and the lemma follows by taking $A=N_{m}^{1/2}$.
\end{proof}

\section{Fluctuations of the leading constant}

\label{sec:leading constant}

\subsection{Some basic observations}

\label{sec:fluct leading basic}

Recall that given $m$ we denoted $\Ec$ to be the set of lattice points
on the circle of radius $\sqrt{m}$, and that we defined the probability measures
$\tau_{m}$ on $\Sc^{1}$ as in \eqref{eq:taum def}.
We may then rewrite $B_{\curve}(\Ec)$
\eqref{def of B in int} as
\begin{equation*}
B_\curve(\vE):=
\int_{\curve}\int_{ \curve}  \int\limits_{\mathcal{S}^{1}} \langle
\theta,\dot{\gamma}(t_{1}) \rangle^{2} \langle
\theta,\dot{\gamma}(t_{2}) \rangle ^{2} d\tau_{m}(\vartheta) dt_{1}dt_{2}.
\end{equation*}
More generally, for any probability measure $\tau$ on $\Sc^{1}$,
invariant w.r.t. $\frac{\pi}{2}$--rotations and the reflection
$(x,y)\mapsto (x,-y)$ we define the number
\begin{equation}
\label{eq:c(mu,gamma) def}
c(\tau,\gamma) = \int\limits_{0}^{L}\int\limits_{0}^{L}
\int\limits_{\mathcal{S}^{1}} \langle
\theta,\dot{\gamma}(t_{1}) \rangle^{2} \langle
\theta,\dot{\gamma}(t_{2}) \rangle ^{2} d\tau(\vartheta)dt_{1}dt_{2} =
\int\limits_{\mathcal{S}^{1}}d\tau(\theta) \left[\int\limits_{0}^{L}
\langle \theta,\dot{\gamma}(t) \rangle ^{2}dt \right]^{2},
\end{equation}
so that, in particular,
\begin{equation}
\label{eq:B=c(tau,gamma)}
B_\curve(\vE) = c(\tau_{m},\gamma).
\end{equation}

The leading constant \eqref{eq:c(mu,gamma) def} is intimately
related with the (weak) limiting angular distribution of lattice
points in $\Ec$. As usual when we deal with convergence of measures,
weak convergence is denoted by ``$\Rightarrow$". Thus if $\{m_{j}
\}$ is a subsequence of energy levels such that
$\tau_{m_{j}}\Rightarrow \tau$ for some probability measure $\tau$
on $\Sc^{1}$ then
$$c(\tau_{m_{j}},\gamma)\rightarrow c(\tau,\gamma).$$
Therefore the variety of limiting values of $B$ is related to the
weak partial limits of $\{\tau_{\eigen}\}$, i.e. probability
measures $\tau$ on $\Sc^{1}$ such that for some subsequence
$\eigen_{j}$ of energy levels, such that
$\tau_{\eigen_{j}}\Rightarrow\tau$. The classification of {\em all}
such measures $\tau$, called {\em attainable}, was first addressed
in ~\cite{KKW}, and was subsequently studied in more detail in
~\cite{KW}. It is well known that the lattice points $\Ec$ are
equidistributed on $\Sc^{1}$ along generic subsequences of energy
levels (see e.g. \cite{FKW}, Proposition 6) in the sense that
$\tau_{m_{j}}\Rightarrow \frac{1}{2\pi}d\theta$ along some density
$1$ sequence $\{m_{j}\}$, and thus, in particular, the normalized arc-length
measure $\frac{1}{2\pi}d\theta$ on $\Sc^{1}$ is attainable. Among
other things it was shown in ~\cite{KW} that for $\tau$ attainable
the value of the Fourier transform $\widehat{\tau}(4)$ attains the
whole interval $[-1,1]$, a fact that is going to be important in the
example considered in section \ref{sec:leading const circ} below.

\subsection{An example: explicit computation of $c(\tau,\gamma)$ for circular arcs}

\label{sec:leading const circ}

Let $\curve$ be the circular arc $$\gamma(t) =
(r\cos(t/r),r\sin(t/r)),$$ $t\in [0,L]$. Here we obtain
after some elementary manipulations
\begin{equation}
\label{eq:c(mu,gamma) circle explicit}
\begin{split}
c(\tau,\gamma) =
\frac{1}{4}L^{2} +\frac{1}{8}r^{2}\sin^{2}(L/r)
+\frac{1}{8}r^{2}\sin^{2}(L/r)\cos(2L/r) \cdot \widehat{\tau}(4),
\end{split}
\end{equation}
where we exploited the $\pi/2$--invariance of $\tau$ to write $\widehat{\tau}(2)=0$.
Since, as it was mentioned in section \ref{sec:fluct leading basic},
all the values of $\widehat{\tau}(4)\in [-1,1]$ are hit by
attainable measures,
the leading constant $4c(\tau,\gamma)-L^{2}$ in \eqref{eq:nodal var
intr asympt} takes all values between
\begin{equation*}
 r^{2}\sin^{4}(L/r) \mbox{ and } r^{2}\sin^{2}(L/r)\cos^{2}(L/r).
\end{equation*}

We may also infer from \eqref{eq:c(mu,gamma) circle explicit} that
if $\gamma$ is a $\frac{1}{8}$-circle plus a multiple of a
quarter-circle ($L/r=\frac{\pi}{4} + k\pi/2$, $k=0,1,2,3$), or a
multiple of a semi-circle ($L/r=\pi,2\pi$), then the leading
constant is independent of $\tau$. For the latter
case the constant {\em vanishes universally}; here the nodal length
fluctuations are of lower
order of magnitude than prescribed by Theorem \ref{thm:nodal var intr asympt}.
The only other case when the leading constant vanishes
occurs for quarter circles plus multiples of semi-circles and
\begin{equation}
\label{eq:tilt Cil measure}
\tau = \frac{1}{4}(\delta_{\pm \pi/4}+\delta_{\pm 3\pi/4})
\end{equation}
the ``tilted Cilleruelo measure" (attainable), name inspired from the ``Cilleruelo measure" ~\cite{KKW,Cil}
\begin{equation}
\label{eq:Cil measure}
\tau = \frac{1}{4}(\delta_{\pm 1}+\delta_{\pm i})
\end{equation}
(when thinking $\Sc^{1}\subseteq\C$); these are excluded from our discussion by
bounding $|\widehat{\tau}(4)|$ away from $\pm 1$ (see e.g. the formulation of Theorem \ref{thm:expected, variance small}).

\subsection{Classification of the leading constants}

By applying the Cauchy-Schwartz inequality on \eqref{eq:c(mu,gamma) def}
it is obvious that for
all $\tau$, $\gamma$, one has $c(\tau,\gamma)\le L^{2}$. A stronger
bound is possible, thanks to the $\pi/2$--rotation invariance of
$\tau$.

We will employ an auxiliary notation in order to rewrite the definition
\eqref{eq:c(mu,gamma) def} of $c(\tau,\gamma)$ in a more useful way
for our purposes.
Given a direction $$\theta=e^{i\vartheta}\in\mathcal{S}^{1}$$ we
denote the $L^{2}$-squared energy of the projection of the tangent
directions of $\gamma$ in the direction $\theta$:
\begin{equation}
\label{eq:A(gamma,theta) def} A(\gamma,\theta):= \int\limits_{0}^{L}
\langle \theta,\dot{\gamma}(t) \rangle ^{2}dt,
\end{equation}
so that we may rewrite \eqref{eq:c(mu,gamma) def} as
\begin{equation}
\label{eq:c(mu,gamma)=int(A^2)} c(\tau,\gamma) =
\int\limits_{\mathcal{S}^{1}} A(\gamma,\theta)^{2}d\tau(\theta).
\end{equation}

\begin{proposition}

\label{prop:c(mu,gamma)=L^2/4}

\begin{enumerate}

\item For all $\tau$ measures on $\mathcal{S}^{1}$, and smooth toral curves $\gamma$ one has
\begin{equation}
\label{eq:c(mu,gamma)<=L^{2}/2} \frac{L^{2}}{4} \le c(\tau,\gamma)\le
L^{2}/2.
\end{equation}

\item The minimum value $$c(\tau,\gamma)= \frac{L^{2}}{4}$$ is attained
for a given measure $\tau$ if and only if for all $\theta$ in the support of $\tau$,
$A(\gamma,\theta)=\frac{L}{2}$.

\end{enumerate}

\end{proposition}

\begin{proof}

\vspace{5mm}

We observe that for $\theta^{\perp}$ a perpendicular direction to
$\theta$ (any of the two),
$$A(\gamma,\theta)+A(\gamma,\theta^{\perp}) = L,$$ from
which it is easy to show that
\begin{equation}
\label{eq:Atheta+Athetaperp>=1/2} \frac{L^{2}}{2} \le
A(\gamma,\theta)^{2}+ A(\gamma,\theta^{\perp})^{2} \le L^{2}.
\end{equation}
We then use the invariance properties of $\tau$ to write
\eqref{eq:c(mu,gamma)=int(A^2)} as
\begin{equation*}
c(\tau,\gamma) = \int\limits_{\mathcal{S}^{1}/i}
2(A(\gamma,\theta)^{2}+ A(\gamma,\theta^{\perp})^{2})d\tau(\theta),
\end{equation*}
where $\mathcal{S}^{1}/i$ is a quarter of the circle identifying
$\vartheta$ and $\vartheta+\pi/2$ of measure
$$\tau(\mathcal{S}^{1}/i)=\frac{1}{4}$$ by the invariance. It then
readily yields via \eqref{eq:Atheta+Athetaperp>=1/2} that
\begin{equation}
\label{eq:c(mu,gamma)>=L^{2}/4} c(\tau,\gamma) =
\int\limits_{\mathcal{S}^{1}/i} 2(A(\gamma,\theta)^{2}+
A(\gamma,\theta^{\perp})^{2})d\tau(\theta)\ge \frac{L^{2}}{4},
\end{equation}
and also \eqref{eq:c(mu,gamma)<=L^{2}/2}.
This concludes the proof of the first statement of the present
proposition, and, in fact, this proof also yields the second one.
\end{proof}

The following corollary from Proposition \ref{prop:c(mu,gamma)=L^2/4} gives the necessary and
sufficient conditions for the leading constant to vanish (equivalently,
for $c(\tau,\gamma)$ to attain its theoretical minimum $c(\tau,\gamma)=\frac{L^{2}}{4}$).
Define the complex number $\mathcal{I}(\gamma)$ as
\begin{equation*}
\mathcal{I}(\gamma) = \int\limits_{0}^{L}e^{2i\varphi(t)}dt =0,
\end{equation*}
where $\dot{\gamma}(t) = e^{i\varphi(t)}$, i.e. $\varphi(t)$ is the angle of
$\dot{\gamma}(t)$ w.r.t. the coordinate axes.

\begin{corollary}
\label{cor:lead non vanish}
\begin{enumerate}

\item
The minimum value $$c(\tau,\gamma)= \frac{L^{2}}{4}$$ is attained universally
(i.e. for all $\tau$),
if and only if
\begin{equation}
\label{eq:int(e^{2iphi})=0}
\mathcal{I}(\gamma) =0,
\end{equation}

\item

If \eqref{eq:int(e^{2iphi})=0} is not satisfied, then the the only
measures $\tau$ where $c(\tau,\gamma)$ may equal $\frac{L^{2}}{4}$ are the Cilleruelo
measure \eqref{eq:Cil measure} and the tilted Cilleruelo
\eqref{eq:tilt Cil measure}; it will occur if and only if
$$\Re\mathcal{I}(\gamma)=\int\limits_{0}^{L}\cos(2\varphi(t))dt =
0  \mbox{ or }
\Im\mathcal{I}(\gamma)=\int\limits_{0}^{L}\sin(2\varphi(t))dt=0$$
respectively.

\end{enumerate}

\end{corollary}

\begin{proof}

Under the notation $\dot{\gamma}(t)
= e^{i\varphi(t)}$ as above,
\begin{equation*}
A(\gamma,\theta) =
\int\limits_{0}^{L}\cos(\vartheta-\varphi(t))^{2}dt = \frac{L}{2} +
\frac{1}{2}\int\limits_{0}^{L}\cos(2(\vartheta-\varphi(t)))dt,
\end{equation*}
and therefore $A(\gamma,\theta) = \frac{L}{2}$ if and only if
\begin{equation*}
\int\limits_{0}^{L}\cos(2(\vartheta-\varphi(t)))dt = 0.
\end{equation*}
Now the latter integral is
\begin{equation*}
\int\limits_{0}^{L}\cos(2(\vartheta-\varphi(t)))dt =
\cos(2\vartheta)\cdot \int\limits_{0}^{L}\cos(2\varphi(t))dt +
\sin(2\vartheta)\int\limits_{0}^{L}\sin(2(\vartheta-\varphi(t)))dt.
\end{equation*}
Thus, if the tuple $(\cos(2\vartheta),\sin(2\vartheta))$ attains at
least two not co-linear values with
$\vartheta\in\operatorname{supp}(\tau)$, it implies that
\begin{equation*}
\int\limits_{0}^{L}\cos(2\varphi(t))dt =
\int\limits_{0}^{L}\sin(2\varphi(t))dt = 0,
\end{equation*}
which is equivalent to \eqref{eq:int(e^{2iphi})=0}; in this case
the constant $c(\tau,\gamma)$ vanishes {\em universally}, i.e.
for all measures $\tau$.

The only two attainable measures that violate the condition of
$$(\cos(2\vartheta),\sin(2\vartheta))$$ attaining at least two not
co-linear values with $\vartheta\in\operatorname{supp}(\tau)$ as
above are Cilleruelo \eqref{eq:tilt Cil measure} and tilted
Cilleruelo \eqref{eq:tilt Cil measure}. Here the condition for
vanishing of the leading constant is
$\int\limits_{0}^{L}\cos(2\varphi(t))dt = 0$ or
$\int\limits_{0}^{L}\sin(2\varphi(t))dt=0$ respectively, as
prescribed.
\end{proof}

The next proposition studies when $c(\tau,\gamma)$ attains the ``theoretical maximum" $\frac{L^{2}}{2}$.

\begin{proposition}

The maximum value $c(\tau,\gamma)=\frac{L^{2}}{2}$ is attained for
$\tau$ the Cilleruelo measure \eqref{eq:Cil measure} and $\curve$ a
straight line parallel to either of the axes, or $\tau$ the tilted
Cilleruelo measure \eqref{eq:tilt Cil measure} and $\curve$ parallel
to $y=\pm x$. Though excluded by Theorem \ref{thm:nodal var intr
asympt}, this could be approximated arbitrarily well by
$c(\tau,\gamma)$ for length-$L$ smooth curves with non-vanishing
curvature.
\end{proposition}

\begin{proof}
By the proof of Proposition \ref{prop:c(mu,gamma)=L^2/4} above
the upper bound in \eqref{eq:c(mu,gamma)<=L^{2}/2} is attained if and only if for
all $\theta\in \operatorname{supp}(\tau)$,
$$A(\gamma,\theta)^{2} + A(\gamma,\theta^{\perp})^{2}=L^2,$$ which happens if and only if for all
$\theta\in\operatorname{supp}(\tau)$ one has $A(\gamma,\theta)=0$ or
$A(\gamma,\theta^{\perp})=0$. Equivalently, for all $\theta \in
\operatorname{supp}(\tau)$ and all $t\in [0,L]$, either $\theta
\perp \dot{\gamma}(t)$ or $\theta^{\perp}\perp\dot{\gamma}(t)$. Thus
there is a ``unique" maximizer for $c(\tau,\gamma)$, where $\tau$ is
an attainable measure and $\gamma$ is a curve, namely the only cases
prescribed in the statement of the present proposition. Since we
exclude the straight lines from our discussion, this is the supremum
rather than maximum.
\end{proof}

\appendix

\section{Non-degeneracy of the covariance matrix}

\label{app:det>0}

In this section we prove Lemma \ref{lem:matrix nongen diag}: given a fixed $0<\epsilon_{0} < 1$ we are to find a
constant $c_{0}=c_{0}(\epsilon_{0})$, so that for all $\eigen$ satisfying $|\tau_{\eigen}|<1-\epsilon_{0}$
(with $\tau_{\eigen}$ defined in \eqref{eq:taum def}), we have $\det{\Sigma(t_{1},t_{2})}>0$ 
(with $\Sigma(t_{1},t_{2})$ given by \eqref{eq:Sigma def})
is strictly positive for $|t_{2}-t_{1}|\le \frac{c_{0}}{\sqrt{m}}$.
Recall that $\mu$ and $\rho$ are given by \eqref{eq: mu def} and \eqref{eq:rho def} respectively
(with $\alpha=2\pi^{2}\eigen$); we have explicitly
\begin{equation}
\label{eq:det(Sigma) expl}
\det\Sigma(t_{1},t_{2}) = \det{A}\cdot \det{\Omega}=(1-r^{2})\cdot (1-r^{2})^{-2}\mu^{2}(1-\rho^{2}) =
(1-r^{2})^{-1}\mu^{2}(1-\rho^{2}).
\end{equation}

As above, $\dot{\gamma}(t)=e^{i\varphi(t)}$, i.e. the vector $\dot{\gamma}(t)$ is a unit vector in the direction $\varphi(t)$,
and $$A(t):=\widehat{\tau_{\eigen}}(4)\cdot \cos(4\varphi(t)).$$
In order to establish the positivity of $\det{\Sigma}(t_{1},t_{2})$ we Taylor expand the expression $\mu^{2}\cdot (1-\rho)^{2}$,
considered as a function of $t_{2}$ and $t_{1}$ constant, around $t_{2}=t_{2}$, as in the following lemma,
with the other term $(1-r^{2})^{-1}$ having been readily expanded \eqref{eq:(1-r^2) taylor diag}.

\begin{lemma}
\label{lem:taylor exp det}
We have 
\begin{equation*}
\begin{split}
\mu^{2}(1-\rho^{2}) &= 
\frac{2}{9}\pi^{14}\eigen^{7}(A(t_{1})-1)(A(t_{1})^{2}-1)(t_{2}-t_{1})^{10}
\\&+O(m^{13/2}(t_{2}-t_{1})^{10}+m^{15/2}(t_{2}-t_{1})^{11}),
\end{split}
\end{equation*}
valid for 
$|t_{2}-t_{1}|\ll \frac{1}{\sqrt{\eigen}}$.
\end{lemma}

\begin{proof}[Proof of Lemma \ref{lem:matrix nongen diag} assuming Lemma \ref{lem:taylor exp det}]

Recall that by \eqref{eq:det(Sigma) expl} we have
\begin{equation*}
\det{\Sigma}=(1-r^{2})^{-1}\cdot \mu^{2}(1-\rho^{2}). 
\end{equation*}
It is obvious from \eqref{eq:(1-r^2) taylor diag} that $(1-r^{2})$ (and hence its 
reciprocal) is strictly positive for $|t_{2}-t_{1}|<\frac{c_{0}}{\sqrt{m}}$ with $c_{0}$ 
depending on $\gamma$ only. Concerning the other factor, we use Lemma \ref{lem:taylor exp det}
to expand 
\begin{equation}
\label{eq:musqr*(1-rho^2) taylor diag}
\begin{split}
&\mu^{2}(1-\rho^{2}) = 
\frac{2}{9}\pi^{14}\eigen^{7}(A(t_{1})-1)(A(t_{1})^{2}-1)(t_{2}-t_{1})^{10}
\\&+O(m^{13/2}(t_{2}-t_{1})^{10}+m^{15/2}(t_{2}-t_{1})^{11}).
\end{split}
\end{equation}
Note that
\begin{equation*}
|A(t_{1})|\le |\widehat{\tau_{\eigen}}(4)| < 1-\epsilon_{0}
\end{equation*}
is bounded away from $1$. That implies that the leading term in 
\eqref{eq:musqr*(1-rho^2) taylor diag}, 
\begin{equation*}
\frac{2}{9}\pi^{14}\eigen^{7}(1-A(t_{1}))(1-A(t_{1})^{2})(t_{2}-t_{1})^{10} \ge \frac{2}{9}\pi^{14}\epsilon_{0}^{3}\cdot\eigen^{7}(t_{2}-t_{1})^{10}
\gg \eigen^{7}(t_{2}-t_{1})^{10},
\end{equation*}
is bigger than the remainder term in \eqref{eq:musqr*(1-rho^2) taylor diag} for $|t_{2}-t_{1}|< \frac{c_{0}}{\sqrt{\eigen}}$
for $c_{0}$ chosen sufficiently small.
\end{proof}

\begin{proof}[Proof of Lemma \ref{lem:taylor exp det}]
We have
\begin{equation}
\begin{split}
\label{eq:musqr(1-r^2) expl}
\mu^{2}(1-\rho^{2}) &= \left( \alpha(1-r^{2})-r_{1}^{2} \right)\left( \alpha(1-r^{2})-r_{2}^{2} \right)
-\left( r_{12}(1-r^{2})+rr_{1}r_{2}\right)^{2}.
\end{split}
\end{equation}

Let $c_{m}=c_{m}(t_{1}):=\frac{\partial^{4} r}{\partial t_{2}^{4}}(t_{1},t_{1})$, 
$e_{m}=e_{m}(t_{1}):=\frac{\partial^{6} r}{\partial t_{2}^{6}}(t_{1},t_{1})$. Using the identities
\begin{equation*}
\cos^{4}{\theta}=\frac{3}{8}+\frac{1}{2}\cos(2\theta)+\frac{1}{8}\cos(4\theta),
\end{equation*}
and
\begin{equation*}
\cos^{6}(\theta)=\frac{5}{16}+\frac{15}{32}\cos(2\theta)+\frac{3}{16}\cos(4\theta)+\frac{1}{32}\cos(6\theta),
\end{equation*}
and $\widehat{\tau_{m}}(k) = 0$ unless $4|k$, by the $\pi/2$ rotation invariance,
we may compute 
\begin{equation}
\label{eq:cm muhat}
\begin{split}
c_{m}&=\frac{(2\pi)^4}{N}\sum\limits_{\mu\in\Ec} (\mu^{t} \cdot \dot{\gamma}(t))^{4} + O(m^{3/2})
\\&=(2\pi)^{4}m^{2}\left(\frac{3}{8}+\frac{1}{8}\widehat{\tau_{m}}(4)\cos(4\phi)\right)+O(m^{3/2}),
\end{split}
\end{equation}
\begin{equation*}
c_{m}' = O(m^{2}),
\end{equation*}
\begin{equation}
\label{eq:em muhat}
\begin{split}
e_{m}&:= -\frac{(2\pi)^6}{N}\sum\limits_{\mu\in\Ec} (\mu^{t} \cdot \dot{\gamma}(t))^{6} + O(m^{5/2})
\\&= -(2\pi)^{6}m^{3}\left( \frac{5}{16}+\frac{3}{16}\widehat{\tau_{m}}(4)\cos(4\phi)  \right) +O(m^{5/2}).
\end{split}
\end{equation}

Let $z:=t_{2}-t_{1}$. Bearing in mind that \eqref{eq:r222 diag=0},
$$\frac{\partial r}{\partial t_{2}}(t_{1},t_{1})=\frac{\partial^{3} r}{\partial t_{2}^{3}}(t_{1},t_{1})=0$$ (cf. 
\eqref{eq:r22(t1,t2)} and \eqref{eq:r222 diag=0}), 
and $$\left |\frac{\partial^{5} r}{\partial t_{2}^{5}}(t_{1},t_{1})\right |=O(m^{2}),$$
we may Taylor expand $r=r(t_{1},t_{2})$ for $t_{1}$ fixed as:
\begin{equation*}
\begin{split}
r&=1-\frac{\alpha}{2}z^{2}+\frac{1}{24}c_{m}(t_{1})z^{4} + 
\frac{1}{720}e_{m}(t_{1})z^{6}+O(m^{2}z^{5}+m^{7/2}z^{7}),
\end{split}
\end{equation*}
where the constant involved in the ``O"-notation depends on $\gamma$ only.
We may differentiate term-wise to obtain
(the terms involving $c_{m}'$, $e_{m}'$ are of smaller order and are absorbed in
the various error terms)
\begin{equation*}
r_{2}=-\alpha z+\frac{1}{6}c_{m}z^{3} + 
\frac{1}{120}e_{m}z^{5}+O\left(m^{2}z^{4}+m^{7/2}z^{6} \right),
\end{equation*}
\begin{equation*}
\begin{split}
r_{1}&=\alpha z -\frac{1}{6}c_{m}z^{3}-
\frac{1}{120}e_{m}z^{5} +
O(m^{2}z^{4}+m^{7/2}z^{6})
\\&=z\left(\alpha -\frac{1}{6}c_{m}z^{2}-
\frac{1}{120}e_{m}z^{4}\right) +
O(m^{2}z^{4}+m^{7/2}z^{6}) ,
\end{split}
\end{equation*}
\begin{equation*}
r_{12} =\alpha-\frac{1}{2}c_{m}z^{2}-
\frac{1}{24}e_{m}z^{4}+
O(m^{2}z^{3}+m^{7/2}z^{5}).
\end{equation*}

Incorporating the above, we have (using $|z|\ll \frac{1}{\sqrt{\eigen}}$ to consolidate the various error
terms throughout)
\begin{equation*}
\begin{split}
1-r^{2}&=(1-r)(1+r) \\&= z^2\left(\frac{\alpha_{2}}{2}-\frac{1}{24}c_{m}z^{2} - \frac{1}{720}e_{m}z^{4}\right)\cdot \left(2-\frac{\alpha}{2}z^{2}+\frac{1}{24}c_{m}z^{4}\right)+ O(m^{2}z^{5}+m^{7/2}z^{7} )
\\&= z^{2}\left(\alpha -\left(\frac{c_{m}}{12}+\frac{\alpha^{2}}{4}\right)z^{2}
+\left(-\frac{e_{m}}{360} + \frac{\alpha}{24}c_{m}\right)z^{4}\right)+ O(m^{2}z^{5}+m^{7/2}z^{7} ),
\end{split}
\end{equation*}
\begin{equation*}
\begin{split}
r_{1}^{2} &= z^{2}\left(\alpha^{2} -\frac{\alpha}{3}c_{m}z^{2}+
\left(\frac{c_{m}^{2}}{36}-\frac{\alpha}{60}e_{m}\right)z^{4}\right) + O(m^{3}z^{5}+m^{9/2}z^{7}),
\end{split}
\end{equation*}
and the same estimate holds for $r_{2}^{2}$;
\begin{equation*}
\begin{split}
\alpha(1-r^{2})-r_{1}^{2} &= z^{4}\left( \frac{\alpha}{4}\left(c_{m}-\alpha^{2}\right)
+\frac{1}{72}\left(e_{m}\alpha+3\alpha^{2}c_{m}- 2c_{m}^{2}\right)z^{2}
 \right)\\&+O(m^{3}z^{5}+m^{9/2}z^{7}),
\end{split}
\end{equation*}
and the same estimate holds for $\alpha(1-r^{2})-r_{2}^{2}$;
\begin{equation}
\label{eq:rhonnorm 1/2}
\begin{split}
&(\alpha(1-r^{2})-r_{1}^{2})\cdot (\alpha(1-r^{2})-r_{2}^{2})\\& =
z^{8}\bigg( \frac{\alpha^{2}}{16}\left(c_{m}-\alpha^{2}\right)^{2}
+\frac{\alpha_{2}}{144}\left(c_{m}-\alpha^{2}\right)\left(e_{m}\alpha+3\alpha^{2}c_{m}-  2c_{m}^{2}\right)z^{2}  \bigg)\\&
+O(m^{6}z^{9}+m^{15/2}z^{11}).
\end{split}
\end{equation}
Continuing,
\begin{equation*}
\begin{split}
&r_{12}(1-r^{2})=z^{2}\left( \alpha-\frac{1}{2}c_{m}z^{2}-
\frac{1}{24}e_{m}z^{4}\right)\times\\&\times\left(\alpha -
\left(\frac{c_{m}}{12}+\frac{\alpha^{2}}{4}\right)z^{2}
+\left(-\frac{e_{m}}{360} + \frac{\alpha}{24}c_{m}\right)z^{4}\right)
+ O(m^{3}z^{5}+m^{9/2}z^{7} ) \\&=
z^{2} \left(\alpha^{2}-
\alpha\left(\frac{7}{12}c_{m}+\frac{\alpha^{2}}{4}\right)z^{2}+
\left( -\frac{2}{45}\alpha e_{m}+\frac{1}{6}\alpha^{2}c_{m}+\frac{c_{m}(t_{1})^{2}}{24}\right)z^{4} \right)
\\&+ O(m^{3}z^{5}+m^{9/2}z^{7} ),
\end{split}
\end{equation*}
\begin{equation*}
\begin{split}
&rr_{1}r_{2} =
-z^{2}\left( 1-\frac{\alpha}{2}z^{2}+\frac{1}{24}c_{m}z^{4}
\right) \cdot
\left(\alpha -\frac{1}{6}c_{m}z^{2}-
\frac{1}{120}e_{m}z^{4}\right)^{2} + O(m^{3}z^{5}+m^{9/2}z^{7})
\\&= -z^{2}\left(\alpha^{2}-\frac{\alpha}{6}(3\alpha^{2}+2c_{m})z^{2}
+\left( \frac{5}{24}\alpha^{2}c_{m}-\frac{\alpha}{60}e_{m}+\frac{1}{36}c_{m}^{2}\right)
z^{4}   \right)  \\&+ O(m^{3}z^{5}+m^{9/2}z^{7})
\end{split}
\end{equation*}
Combining the last couple of estimates we obtain:
\begin{equation*}
\begin{split}
&r_{12}(1-r^{2})+rr_{1}r_{2} =z^{4}\left(
\frac{\alpha}{4}(\alpha^{2}-c_{m} + \left( -\frac{1}{36}\alpha e_{m}
-\frac{1}{24}\alpha^{2}c_{m} + \frac{1}{72}c_{m}^{2}\right)z^{2}   \right)
\\&+ O(m^{3}z^{5}+m^{9/2}z^{7}),
\end{split}
\end{equation*}
and
\begin{equation*}
\begin{split}
&(r_{12}(1-r^{2})+rr_{1}r_{2})^{2}\\& =\alpha(\alpha^{2}-c_{m})z^{8}\left(
\frac{\alpha}{16}(\alpha^{2}-c_{m})+
\frac{1}{144}\left( -2\alpha e_{m}
-3\alpha^{2}c_{m} + c_{m}^{2}\right)\cdot z^{2}   \right)
\\&+ O(m^{6}z^{9}+m^{15/2}z^{11}).
\end{split}
\end{equation*}

Finally using the latter estimate with \eqref{eq:rhonnorm 1/2} we obtain (the
term corresponding to $z^{8}$ cancels out precisely, and by the non-negativity
the Taylor series necessarily starts from an even power)
\begin{equation}
\label{eq:musqr*(1-rho^2)}
\begin{split}
&\left( \alpha(1-r^{2})-r_{1}^{2} \right)\left( \alpha(1-r^{2})-r_{2}^{2} \right)
-\left( r_{12}(1-r^{2})+rr_{1}r_{2}\right)^{2}
\\&= \frac{\alpha}{144}(\alpha^{2}-c_{m})(c_{m}^{2}+\alpha e_{m})z^{10}
+O(m^{6}z^{9}+m^{15/2}z^{11}).
\end{split}
\end{equation}

Note that by \eqref{eq:alpha2 explicit}, \eqref{eq:cm muhat} and \eqref{eq:em muhat} we have
\begin{equation*}
\alpha^{2}-c_{m}=2\pi^{4}m^{2}\left( \widehat{\tau_{m}}(4)\cos(4\varphi)-1  \right) + O(m^{3/2}),
\end{equation*}
and
\begin{equation*}
c_{m}^{2}+\alpha e_{m} = 4\pi^{8}m^{4}(\widehat{\tau_{m}}(4)^{2}\cos^{2}(4\varphi)-1)+O(m^{7/2}),
\end{equation*}
so that, bearing in mind \eqref{eq:musqr(1-r^2) expl}, \eqref{eq:musqr*(1-rho^2)} is
\begin{equation}
\label{eq:musqr(1-rho^2) taylor}
\begin{split}
&\mu^{2}(1-\rho^{2})= \frac{2}{9}\pi^{14}\eigen^{7}(A(t_{1})-1)(A(t_{1})^{2}-1)(t_{2}-t_{1})^{10}
+O(m^{6}z^{9}+m^{15/2}z^{11});
\end{split}
\end{equation}
this is almost identical to the statement of the present lemma, except that we have to improve
the error term.
To this end we observe that since, in light of \eqref{eq:det(Sigma) expl}, the expression on the l.h.s.
of \eqref{eq:musqr(1-rho^2) taylor} is non-negative, the Taylor expansion on the r.h.s. of \eqref{eq:musqr(1-rho^2) taylor} 
is guaranteed to begin with an even power of $z$. Hence the first error term
$O(m^{6}z^{9})$ is $O(m^{13/2}z^{10})$
(recall that this expansion is valid for $|t_{2}-t_{1}|\ll \frac{1}{\sqrt{m}}$).
\end{proof}

\end{document}